\documentclass{article}

\usepackage{graphicx}%
\usepackage{multirow}%
\usepackage{amsmath,amssymb,amsfonts}%
\usepackage{amsthm}%
\usepackage{mathrsfs}%
\usepackage[title]{appendix}%
\usepackage{xcolor}%
\usepackage{textcomp}%
\usepackage{manyfoot}%
\usepackage{booktabs}%
\usepackage{algorithm}%
\usepackage{algorithmicx}%
\usepackage{algpseudocode}%
\usepackage{listings}%
\usepackage{geometry}
\usepackage{graphicx} % Required for inserting images
\usepackage{hyperref}
\usepackage{amsmath}
\usepackage{amssymb}
\usepackage{amsthm}
\usepackage{float}
\usepackage{color}
\usepackage{xcolor}
\usepackage{tabularx}
\usepackage{array}
\usepackage{enumitem}
\usepackage{cleveref} 
\usepackage{algorithm}
\usepackage{algpseudocode}
\usepackage{mathtools}
\usepackage{stmaryrd}
\usepackage{booktabs}
\usepackage[normalem]{ulem}
\usepackage{pifont}

\newcommand{\mbf}[1]{\mathbf{#1}}

\newcommand*{\cA}{\mathcal{A}}
\newcommand*{\cB}{\mathcal{B}}
\newcommand*{\cC}{\mathcal{C}}
\newcommand*{\cD}{\mathcal{D}}

\newcommand*{\cH}{\mathcal{H}}

\newcommand*{\cI}{\mathcal{I}}

\newcommand*{\cP}{\mathcal{P}}
\newcommand*{\cQ}{\mathrm{Q}}

\newcommand*{\cS}{\mathcal{S}}
\newcommand*{\cT}{\mathcal{T}}

\newcommand*{\cX}{\mathcal{X}}
\newcommand*{\cY}{\mathcal{Y}}

\newcommand*{\N}{\mathbb{N}}

\newcommand*{\RR}{\mathbb{R}}

\newcommand*{\NN}{\mathbb{N}}

\newcommand{\graph}{\mathrm{graph}\ }
\newcommand{\argmin}{\arg\min}

\newcommand{\tupfuncN}{(P_i)_{i = 1}^N}
\newcommand{\firstvar}{x}
\newcommand{\secondvar}{y}

\newcommand{\SA}{\mathcal{SA}}
\newcommand{\LSC}{\mathcal{LSC}}

\newcommand{\BC}{\mathcal{CB}}
\newcommand{\PP}{\mathcal{PP}}

\newcommand{\property}{\texttt{boxtype}}
\newcommand{\unbounded}{\text{unbound}}
\newcommand{\bounded}{\text{bounded}}

\newcommand{\coco}{\text{cc}}
\newcommand{\conv}[1]{\text{Conv}(#1)}
\newcommand{\existleveltwo}{\Sigma_2^p}
\newcommand{\dom}{\mathrm{dom}}
\newcommand{\valoptim}{\varphi_o}
\newcommand{\valpessim}{\varphi_p}
\newcommand{\valbilevel}{\varphi}

\theoremstyle{plain}
\newtheorem{theorem}{Theorem}[section]

\newtheorem{lemma}[theorem]{Lemma}

\newtheorem{corollary}[theorem]{Corollary}
\newtheorem{proposition}[theorem]{Proposition}

\newtheorem*{theorem*}{Theorem}
\newtheorem*{proposition*}{Proposition}

\theoremstyle{definition}
\newtheorem{definition}[theorem]{Definition}
\newtheorem{remark}[theorem]{Remark}

\newtheorem{problem}[theorem]{Problem}

\newtheorem*{definition*}{Definition}

\definecolor{linkcolor}{RGB}{83,83,182}
\hypersetup{
    colorlinks=true,
    citecolor=linkcolor,
    linkcolor=linkcolor
}

\title{Geometric and computational hardness\\of  bilevel programming}
\author{Jérôme Bolte\thanks{Toulouse School of Economics, University of Toulouse Capitole, Toulouse, France.} \and Tùng Lê \footnotemark[1] \and Edouard Pauwels\footnotemark[1] \and  Samuel Vaiter \thanks{CNRS \& Université Côte d'Azur, Laboratoire J. A. Dieudonné. Nice, France.} }

\begin{document}
\maketitle
\begin{abstract}
    %Bilevel programming is a general framework which encompasses many problems of interest in a diversity of applications. Yet it is a challenging problem class as the theory and practice of bilevel programming lack important assets such as versatile optimality conditions and general purpose solvers. 
    %After observing that
    We first show a simple but striking result in bilevel optimization: unconstrained $C^\infty$ smooth bilevel programming is as hard as general extended-real-valued lower semicontinuous minimization.
    We then proceed to a worst-case analysis of box-constrained bilevel polynomial optimization.  We show in particular that any extended-real-valued semi-algebraic function, possibly non-continuous, can be expressed as the value function of a polynomial bilevel program. Secondly, from a computational complexity perspective, the decision version of polynomial bilevel programming is one level above NP in the polynomial hierarchy ($\existleveltwo$-hard). Both types of difficulties are uncommon in non-linear programs for which objective functions are typically continuous and belong to the class NP. These results highlight the irremediable hardness attached to general bilevel optimization and the necessity of imposing some form of regularity on the lower level.
\end{abstract}
% Title propositions, add your favourite to the list:
% \begin{itemize}
%     \item Nonconvex bilevel optimization
%     \item Hardness of polynomial bilevel-optimization
%     \item \jer{On the hardness of bilevel optimization}
%     \item Polynomial bilevel optimization is hard.
%     \item A worst-case perspective of nonconvex bilevel optimization.
%     \item Geometric and computational hardness of polynomial bilevel programming.
%     \item \jer{Geometric and computational hardness of (nonconvex) bilevel programming.}
%     \item Geometric and computational hardness of nonconvex polynomial bilevel programming. (\TL{I add nonconvex since this keyword is fairly common when we search for papers.)}
% \end{itemize}

\section{Introduction}
This paper is concerned with the bilevel programming problem \eqref{eq:simple-bilevel-poly-optim}, that is formulated as follows:
\begin{equation}
    \label{eq:simple-bilevel-poly-optim}
    \tag{BP}
    \begin{aligned}
        \underset{\firstvar \in \cX}{\text{Minimize}} & \quad P(\firstvar,\secondvar)\\
        \text{s.t.} &\quad \secondvar \in \Theta(\firstvar) :=  \underset{\secondvar \in \cY}{\argmin}\; Q(\firstvar,\secondvar)
    \end{aligned}
\end{equation}
where $P, Q: \RR^n \times \RR^m \to \RR$ are called upper-level and lower-level functions respectively. Similarly, $\firstvar$ and $\secondvar$ are called upper-level and lower-level variables. We will focus on box or convex compact semi-algebraic lower-level feasible set $\cY$, independent of the upper variable $\firstvar$\footnote{Due to this independence, the formulation \eqref{eq:simple-bilevel-poly-optim} is sometimes called \emph{simple bilevel optimization} \cite{jiawang2017bilevel,guihua2014solving}. The general bilevel formulation might replace $\cY$ by $\cY(\firstvar)$ where the feasible set depends on the upper variable $\firstvar$.}, and polynomial $P,Q$.

Since $Q(\firstvar, \cdot)$ might admit more than one minimizer, if $y$ is not considered as a free optimization variable, the problem \eqref{eq:simple-bilevel-poly-optim} may be ill-posed. There are two major approaches to choose a minimizer in $\Theta(\firstvar)$: optimistic and pessimistic\footnote{Apart from these two models, there exist others such as bilevel optimization under uncertainty \cite{beck2021robust,beck2023survey,Aboussoror1995StrongweakSP}. A complete presentation of these models is, however, out of the scope of this paper.}. They yield two different \emph{value functions}:

\noindent\begin{minipage}{.5\linewidth}
\begin{equation}
    \label{eq:optimistic-bilevel}
    \tag{BP-O}
    \valoptim\colon \firstvar \mapsto \inf_{\secondvar \in \Theta(x)} \; P(\firstvar,\secondvar)
\end{equation}
\end{minipage}%
\begin{minipage}{.5\linewidth}
\begin{equation}
    \label{eq:pessimistic-bilevel}
    \tag{BP-P}
    \valpessim \colon x \mapsto \sup_{\secondvar \in \Theta(\firstvar)} \; P(\firstvar,\secondvar)
\end{equation}
\end{minipage},
where $\Theta(\firstvar)$ is defined as in \eqref{eq:simple-bilevel-poly-optim}. %The use of ``$\min/\max$'' to define value functions in \eqref{eq:optimistic-bilevel} and \eqref{eq:pessimistic-bilevel} implicitly means that $\Theta(\firstvar)$ is non-empty for all $\firstvar \in \cX$ and $P(\firstvar,\cdot)$ attains its min/max values in $\Theta(\firstvar)$. This is formalized in \Cref{assumption:well-definedness}. 
We distinguish $\valoptim$ (resp. $\valpessim$) the optimistic (resp. pessimistic) value function, where $\valoptim(\firstvar) = + \infty$ (resp. $\valpessim(\firstvar) = - \infty$), if $\Theta(\firstvar) = \emptyset$.  We emphasize that the term ``value function'' is different from the ``lower-level value function'', which has been used in the literature to construct algorithms, starting with \cite{ye1995optimality}. For simplicity, we will focus on the optimistic formulation of the bilevel problem, for the rest of this paper, the term ``bilevel problem'' in \eqref{eq:simple-bilevel-poly-optim} denotes the optimistic formulation \eqref{eq:optimistic-bilevel} and $\valbilevel = \valoptim$ denote the corresponding value function. This amounts to considering that the minimization in \eqref{eq:simple-bilevel-poly-optim} also takes place over the variable $y$. The majority of our results have a counterpart for the pessimistic formulation, see \Cref{subsection:summary_results}, with lower semicontinuity is replaced by upper semicontinuity.

%\TL{In this paper, our main focus is polynomial bilevel optimization, i.e., $P,Q$ are polynomials, and $\cY$ is a simple feasible set described by polynomial equalities and inequalities. We explain our motivation for such an emphasis in the following.}
\subsubsection*{Motivations and objectives}
%Our work is originally motivated by the applications of bilevel optimization. 
Historically, bilevel optimization has been used to address problems in economics, chemistry, optimal control, and decision-making. For some insights into these traditional applications, we refer readers to \cite[Chapter 1]{dempe2002foundations}, \cite{sinha2018review,cerulli:tel-03587548}. Recently, there has been a regain of interest in bilevel optimization among the machine learning community due to its applications in hyper-parameters tuning and meta-learning \cite{implicitMAML,Franceschi2018BilevelPF}. 

A large body of work on bilevel programming, especially in the context of machine learning \cite{Dagreou2022SABA,Dagreou2023SRBA,Liu2022investigating,gould2016differentiating,ablin2020super-efficiency,Blondel2021EfficientAM,arbel2022amortized,ji2020BilevelOC}, requires strong convexity of the lower-level problem to design scalable and provably convergent algorithms. This is a very favorable situation as the bilevel constraint is essentially equivalent to a {\em qualified} equality constraint, i.e., a manifold-like constraint. However, for many real-world machine learning problems (e.g., meta-learning and hyper-parameters optimization \cite{Franceschi2018BilevelPF}), the lower-level problem is not strongly convex and even non-smooth and nonconvex. To address the challenges in those situations, existing works \cite{liu2021towards,liu2021valuefunction,arbel:hal-03869097,guihua2014solving,Mordukhovich2020,ye1995optimality,dempe2012sensitivity,liu2020generic} have proposed various regularity, variational and ad-hoc assumptions on $P,Q,\cY$. These assumptions are, however, difficult to verify \textit{a priori} and may fail to hold for typical bilevel problems \cite[Section 3.2]{hendrion2011calmness}.  

In this work, we adopt a worst-case approach to explore the hardness of bilevel optimization and identify the class of functions that can be represented using a general bilevel problem (\eqref{eq:simple-bilevel-poly-optim}, especially in the setting where the lower-level programs are smooth but possibly nonconvex. Our analysis highlights the difficulties of general bilevel optimization, {in particular by thoroughly studying the pathologies of polynomial programming.}.%, and the necessity of qualification conditions. %Our result demonstrates that the worst-case of bilevel optimization with smooth, nonconvex lower-level problem is already extremely degenerate, suggesting that it can get worse if 

\subsubsection*{Pathological value functions are omnipresent}

Bilevel programming falls in the general framework of parametric optimization with constraints parameterized by the upper variable $x$. In this context, favorable situations, include cases where the constraint $y \in \Theta(x)$ can be equivalently described with a well-behaved equation or defines a smooth or regular mapping. This is typical in the smooth setting if the lower-level is strongly convex, which results in a smooth mapping $\Theta$.
A typical difficulty of general bilevel programs is that the resulting value function does not inherit the smoothness properties of its defining data $P$ and $Q$ as the argmin mapping  $\Theta$ may be poorly structured ---precisely because it corresponds to a critical set, here an argmin. The following proposition illustrates this behavior.
\begin{theorem}[Lower semicontinuous minimization problems are $C^\infty$  bilevel problems]
    \label{prop:reduction-lsc}
    Any proper lower semicontinuous function $f:\RR^n \to \RR\cup \{+\infty\}$ whose domain $\dom(f):= \{x \in \RR^n \mid f(x) \in \RR\}$ is closed, is the value function of a bilevel problem with $C^\infty$ smooth upper and lower levels $P,Q$ and $\cY = \RR^{3}$.  
\end{theorem}
%\jer{Jérôme suggère d'enlever ou de déplacer tout ce qui suit jusqu'à la partie rouge. Pour améliorer la lisibilité du théorème on peut enlever au dessus pessimists et dire dessous qu'on peut faire pareil. Il faut en effect réaliser que l'intro atteint les 4 pages ce qui est très long et peut enterrer notre papier. Les rapporteurs ont probablement raison d'appeler à la simplification} 
%In particular, any full domain lower semicontinuous function $f: \RR^n \to \RR$ is the value function of a bilevel problem with $C^\infty$ smooth upper and lower levels. %sam: j'ai rajoute ce commentaire trivial mais j'aime bien l'enonce sans discussion sur le domaine.
%The proof of \Cref{prop:reduction-lsc} is given in \Cref{appendix:regularity-not-good-enough}. In particular, \Cref{prop:reduction-lsc} 
%In particular, \emph{any method resembling a gradient algorithm on the value function may encounter insurmountable difficulties \textbf{even} if the problem data is arbitrarily smooth}.
{Further comments and the proof of \Cref{prop:reduction-lsc}, are postponed to \Cref{appendix:regularity-not-good-enough} in the appendix.  The main argument relies on the characterization of closed sets as the zero set of smooth functions due to Whitney (\Cref{theorem:whitney-representation}). }%The implication of \Cref{prop:reduction-lsc} is also discussed more thoroughly in \Cref{appendix:regularity-not-good-enough}.}  
{Let us mention, however, that the pathology of bilevel value functions can actually be much worse. %:  \Cref{prop:reduction-lsc} only shows an inclusion, and there may exist $C^\infty$ bilevel problems which lead to more serious pathologies. 
For example, it is known in probability theory that minimization processes can even destroy Borel measurability, see for example \cite[Chapter 7]{bertsekas1996stochastic}. One can use this to show that certain bilevel optimization problems are equivalent to a constrained optimization problem over a set that is not Borel measurable. }% In the same spirit, one could also consider weaker regularity for $P,Q$ (e.g., locally Lipschitz, analytical, lower-semicontinuous etc.).}%Yet lower semicontinuity remains central in optimization as it is one of the weakest sufficient condition to enforce existence of minimizers, and we leave a more precise analysis of an exact characterization of value functions representable with smooth bilevel programs for future research.}

\subsubsection*{Contributions} 
In light of the preceding examples, it is natural to turn towards more rigid classes $\{P,Q,\cY\}$ to hope for better results without compromising the model's applicability to concrete problems. Polynomial functions and sets constitute one of the simplest instances of such classes: they have a strict algebraic nature, yet they cover a wealth of concrete real-world applications. %The box-like choice of $\cY$ is for an easier analysis but we will see that even with this simplest constraint, polynomial bilevel optimization is already complex (cf. \Cref{theorem:closed=SA,theorem:compact=lsc+bc,theorem:hardness2}).
We actually pertain to box-constrained polynomial bilevel optimization, i.e., $P, Q$ are polynomials and $\cY$ has a box-like form $[a_1,b_1] \times \ldots [a_m, b_m]$ where $a_i, b_i \in \RR \cup \{\pm \infty\}, i = 1, \ldots, m$, independent of the upper level variable $x$. %Our choice of polynomials is motivated by the fact that it is a natural candidate to study bilevel optimization without strong convexity in lower-level problems. Indeed, bilevel optimization with usual regularity assumptions (such as smooth - $C^\infty$ functions $P, Q$) remains challenging for both algorithmic design and analysis. In fact, one can reduce the problem of optimizing a lower semicontinuous (resp. upper semicontinuous) to the bilevel problem with smooth functions $P,Q$. \TL{This claim can be formally stated as follows:}
We investigate the worst-case behavior of such  bilevel polynomial optimization in two different ways:
\begin{enumerate}
    \item \textbf{Geometric complexity}: Solving a bilevel problem is equivalent to optimizing its value function $\valbilevel$. Loosely speaking, we show that box-constrained polynomial bilevel programming is equivalent to the optimization of an arbitrary real semi-algebraic function (non necessarily continuous). In other words, arbitrary polynomial bilevel programming is not tractable. We actually provide sharp characterizations of the class of value functions of \eqref{eq:simple-bilevel-poly-optim} in various scenarios, these results are given in Sections \ref{sec:geometricHardness}.
    \item \textbf{Computational complexity}: We investigate \eqref{eq:simple-bilevel-poly-optim} along the angle of the classical computational complexity theory \cite{arora2009computational}. Our second main result asserts that the decision version of bilevel polynomial optimization is $\existleveltwo$-hard. This implies that bilevel polynomial optimization is more challenging than many $\textbf{NP}$-complete problems.\footnote{Assuming that the polynomial hierarchy does not collapse on the first level, for which a broad consensus exists.}% A more detailed discussion about this and other related results will be provided latter in the introduction as well as in \Cref{subsec:premilinary-poly-hierarchy}.}
\end{enumerate}

\subsection*{A brief discussion about the connection to existing work}
%\jer{Les deux paragraphs qui suivent sont très long alors qu'on veut juste dire : en général les algo vont pas marcher et une façon d'éviter cela c'est les QC, ce qui est au fond standard}
{In order to avoid facing the  pathologies we present, many work rely on strong or complex assumptions on the inner level. We evoke this briefly below; a more detailed discussion is provided in \Cref{appendix:existing-works}. \\
 To avoid ``monsters" induced by the argmin operator, there is quite important literature which focuses on the establishment of necessary conditions for locally optimal solutions. They are connected, implicitly or explicitly, to the so-called qualification conditions, see e.g.,   \cite{ye1995optimality,dempe2012sensitivity,Mordukhovich2020,hendrion2011calmness,chen1995nonlinear}.  Algorithms for bilevel optimization, especially when the lower-level problem is nonconvex, are quite demanding in terms of regularity conditions; some examples can be found in \cite{liu2021towards,liu2021valuefunction,arbel:hal-03869097,liu2020generic, guihua2014solving}. Some of these works are tailored for polynomial bilevel optimization \cite{jeyakumar2016convergent,jiawang2017bilevel} but mostly through algebraic techniques. In any case,  theoretical guarantees for solution methods can only be established under (very) strong qualification assumptions or for relaxed problems (e.g., $y$ is only a stationary point of $Q(x, \cdot)$). This phenomenon is, of course, consistent with the pathologies highlighted in \Cref{prop:reduction-lsc} and research on the computational complexity of bilevel optimization \cite{jeroslow1985polynomial,bard1991property,ben-ayed1990computational,lodi2013bilevel,caprara2014study}.}

\section{Geometric hardness of polynomial bilevel optimization}
\label{sec:geometricHardness}

Solving a bilevel problem is equivalent to optimizing the corresponding value function, $\valbilevel$, which depends solely on the upper-level variables. In this section, we study the complexity of the bilevel polynomial optimization problem by investigating the sets of value functions that can be ``expressed'' via polynomial bilevel formulations. 

% To ensure that $\valoptim$ or $\valpessim$ are well-defined, we will work under the following assumption:
% \begin{assumption}
%     \label{assumption:well-definedness}
%     The functions $P,Q$ and lower-level feasible set $\cY$ satisfy that the lower-level solution set $\Theta(\firstvar)$ (cf. \eqref{eq:simple-bilevel-poly-optim}) is non-empty for all $\firstvar$ and $P(\firstvar,\cdot)$ attains its min/max values in $\Theta(\firstvar)$.
% \end{assumption}

% We remark that if $P,Q$ are continuous with $\cY$  nonempty and compact, then \Cref{assumption:well-definedness} is automatically satisfied. We focus on subclasses of polynomial bilevel problems with variations of boundedness and convexity criteria. 

First, we distinguish situations for which the feasible set of the lower-level optimization problem, $\cY$, is a box of the form $[a_1,b_1] \times \ldots [a_m, b_m]$ which can be bounded or unbounded. More specifically, we say that $\cY$ is an \emph{unbounded box} (or $\cY$ is \emph{unbounded}) if there exists an index $1 \leq i \leq m$ such that $a_i = -\infty$ or $b_i = +\infty$. Otherwise, $\cY$ is a \emph{bounded box} (or simply, \emph{bounded}). It is worth mentioning that, since we are investigating worst case properties, the restriction to box type constraints is actually a strength and our results hold for more general constraint sets. 

Second, we restrict our attention to convex lower-level problems. This means that the lower-level objective $Q(\firstvar, \secondvar)$ is convex with respect to its second argument $\secondvar$ and the constraint set $\cY$ is also convex. This represents a natural intermediate situation between strongly convex lower levels and general nonconvex ones. In this setting, we make a distinction between bounded and unbounded box on the one hand, as well as general compact convex sets on the other hand.

%It is worth emphasizing that our results do not contradict the existence of efficient algorithms \cite{Dagreou2022SABA,Dagreou2023SRBA,Liu2022investigating,gould2016differentiating,ablin2020super-efficiency,Blondel2021EfficientAM} in the literature because they are proposed for the lower-level \emph{strongly} convex problems. 

Our analysis provides a sharp characterization of the class of functions which can be represented as the value function of polynomial bilevel programs: we show that for the most natural types of bilevel programs, this class is the largest possible. This is described in \Cref{subsection:closed-results}. Furthermore, our analysis shows that enforcing convexity of the lower-level problems does not reduce much the corresponding class of polynomial bilevel value functions. %This is in stark contrast with strong convexity for which there is an equivalent single level equality constrained formulation with algorithmic solutions \cite{Dagreou2022SABA,Dagreou2023SRBA,Liu2022investigating,gould2016differentiating,ablin2020super-efficiency,Blondel2021EfficientAM}. 
This is discussed in \Cref{subsection:compact-results}, which is more technical than the general case for which our results are sharper. In the convex setting, we distinguish between box-constrained convex lower level, which we relate to the class of piecewise polynomial functions, and general convex compact semi-algebraic set for which we obtain a sharper characterization.

The rest of this section is organized as follows. We start in \Cref{sec:preliminaries} by providing readers with preliminaries on the mathematical tools required for our analysis: classical results from semi-algebraic geometry and bilevel optimization. The representation results for general polynomial bilevel problems are given in \Cref{subsection:closed-results} and the specification to convex lower-level problems is described in \Cref{subsection:compact-results}. A summary of all the results of this section is provided in \Cref{subsection:summary_results}.

\subsection{Preliminaries}
\label{sec:preliminaries}
First, we will recall the basic definitions and results of semi-algebraic geometry and their consequences for polynomial bilevel programming. Second, we remind the readers of a simplified version of Berge's Maximum Theorem \cite[Section 6.3]{maximum_theorem} and its consequence on semicontinuity properties of bilevel programs. These classical results play a central role in our analysis, and we provide a detailed account for completeness. The informed reader may skip this subsection.

\subsubsection*{Semi-algebraic geometry and polynomial bilevel value functions} 
We provide some definitions and facts about semi-algebraic geometry. For an excellent exhibition of this subject, we refer readers to \cite{realgeoalg}.

\begin{definition}[Semi-algebraic sets and functions]
    \label{def:semi-algebraic-set}
    A set $S \in \RR^n$ is a \emph{basic} semi-algebraic set if it has the form:
    \begin{equation*}
        \cS = \{x \in \RR^n \mid P(x) = 0 \text{ and } Q_{j}(x) > 0, j \in J\},
    \end{equation*}
    where $J$ is a finite index set the functions $P, Q_{j}, j \in J$ are polynomials. A semi-algebraic set is a finite union of basic semi-algebraic sets.  
    
    A function $f: U \subseteq \RR^n \to \RR^m$ is a semi-algebraic function if $\graph f:=\{(x,f(x)) \in \RR^n \times \RR^m \mid x \in U\}$ is a semi-algebraic set of $\RR^{n + m}$. 
\end{definition}

One can even define a semi-algebraic set as a \emph{disjoint union} of finite basic semi-algebraic sets (and not just a \emph{finite union} as in \Cref{def:semi-algebraic-set}). In order to be self-contained, we provide a proof for this claim in \Cref{appendix:alternative-def}. In the following, we also consider a definition for extended-real-valued semi-algebraic functions. 

\begin{definition}[Extended-real-valued semi-algebraic functions]
    \label{def:extended-semi-algebraic-function}
    A function $f: \RR^n \to \RR \cup \{\pm \infty\}$ is an extended-real-valued semi-algebraic function if three sets $\dom(f) := \{x \in \RR^n \mid f(x) \in \RR\}, \dom(f)^+:= \{x \in \RR^n \mid f(x) = +\infty\}$ and $\dom(f)^-:= \{x \in \RR^n \mid f(x) = -\infty\}$ are semi-algebraic and the function $f|_{\dom(f)}$ (the restriction of $f$ to $\dom(f)$) is semi-algebraic (cf. \Cref{def:semi-algebraic-set}). The class of all extended-real-valued semi-algebraic functions is denoted by $\SA$. For a semi-algebraic function $f \colon \RR^p \to \RR$, we have  $\dom(f)^+ = \dom(f)^- = \emptyset$, such functions are called real-valued by oposition to extended-real-valued.
\end{definition}

 A fundamental result in semi-algebraic geometry is the Tarski-Seidenberg theorem, proving the stability of semi-algebraic sets under projection (and consequently, first-order logic).

\begin{theorem}[Tarski-Seidenberg theorem]
    \label{theorem:projection-theorem}
    Let $S \subseteq \RR^n$ be a semi-algebraic set, and $\pi: \RR^n \to \RR^{n-1}$ be the projection onto the first $n-1$ coordinates. Then $\pi(S) = \{x \in \RR^{n - 1} \mid \exists y \in \RR, (x,y) \in S\}$ is also semi-algebraic.
\end{theorem}

Besides general semi-algebraic sets, closed semi-algebraic sets have a more particular form, given in \Cref{lemma:closed-semi-algebraic-set}.
\begin{proposition}[Characterization of closed semi-algebraic sets {\cite[Exercise 2.5.7]{benedetti1990real}}]
    \label{lemma:closed-semi-algebraic-set}
    Every \emph{closed} semi-algebraic set $S$ in $\RR^n$ can be represented in the form:
    \begin{equation}
        \label{eq:form-of-closed-semi-algebraic-sets}
        S = \bigcup_{i \in I} {\cS_i \qquad \text{ where } \qquad \cS_i :=} \bigcap_{j \in J} \{x \in \RR^n \mid P_{ij}(x) \geq 0\},
    \end{equation}
    $P_{ij}$ are polynomials and $I,J$ are finite index sets.
\end{proposition}

Finally, we recall a result of the growth of a semi-algebraic function.
\begin{proposition}[Growth of semi-algebraic functions {\cite[Section $4.12$]{o-minimal-structures}}]
    \label{lemma:growth-semi-algebraic-functions}
    For every semi-algebraic function $f:\RR^n \to \RR$, there exists a natural number $N$ and a constant $C$ such that $|f(x)| \leq \|x\|^N$ for all $x, \|x\| \geq C$ where $\|\cdot\|$ indicates the Euclidean norm.  
\end{proposition}

We conclude this section with an application of the Tarski-Seidenberg theorem to justify that value functions of polynomial bilevel problems are semi-algebraic (possibly with extended-real values $\pm \infty$). Consider bilevel optimization: given $P, Q, \cY$, there is a partition of $\RR^n$ into three disjoint sets $\dom(\valbilevel)$, $\dom(\valbilevel)^+$, and $\dom(\valbilevel)^-$ defined as follows:
\begin{enumerate}
    \item $\dom(\valbilevel):= \{\firstvar \in \RR^n \mid \exists \delta \in \RR, \forall \secondvar \in \Theta(\firstvar), \delta \leq P(\firstvar, \secondvar) \text{ and } \forall \epsilon > 0, \exists \secondvar' \in \Theta(\firstvar), \delta + \epsilon > P(\firstvar, \secondvar')\}$: if $\firstvar \in \dom(\valbilevel)$, then $\valbilevel(\firstvar) \in \RR$. 
    \item $\dom(\valbilevel)^+ := \{\firstvar \in \RR^n \mid \nexists \secondvar \in \Theta(\firstvar)\}$: if $\firstvar \in \dom(\valbilevel)^+$, then $\valbilevel(\firstvar) = +\infty$. 
    \item $\dom(\valbilevel)^- := \{\firstvar \in \RR^n \mid \forall \delta \in \RR, \exists \secondvar \in \Theta(\firstvar), P(\firstvar, \secondvar) < \delta\}$: if $\firstvar \in \dom(\valbilevel)^-$, then $\valbilevel(\firstvar) = -\infty$.
\end{enumerate}
where $\Theta(\firstvar) = \{y \in \cY \mid \forall \secondvar' \in \cY, Q(\firstvar,\secondvar) \leq Q(\firstvar,\secondvar')\}$ is also a first-order logic expression (thus, semi-algebraic). Thus, all three sets are semi-algebraic. Moreover, the graph of the restriction of $\valbilevel$ to its domain $\dom(\valbilevel)$ is given by:
\begin{equation*}
    \{(\firstvar,P(\firstvar,\secondvar)) \in \RR^n \times \RR \mid \firstvar \in \dom(\valbilevel), \secondvar \in \Theta(\firstvar) \text{ and }  \forall \secondvar' \in \Theta(\firstvar), P(\firstvar,\secondvar) \leq P(\firstvar,\secondvar')\}. 
\end{equation*}
Using Tarski-Seidenberg quantifier elimination, we have the following proposition:
\begin{proposition}[Semi-algebraicity of value functions]
    \label{cor:upper-bound-nonconvex-class}
    The value functions of any polynomial bilevel optimization problem with box constraint is semi-algebraic.%, i.e.:
    %\TL{All function in $\cP^{\optim}_{\unbounded}$ (resp. $\cP^{\pessim}_{\unbounded}$) are semi-algebraic, i.e.}
    % \begin{equation*}
    %     \cP^{\formule}_{\property} \subseteq \SA \qquad \qquad \text{and} \qquad \qquad \cC^{\formule}_{\property} \subseteq \SA,
    % \end{equation*}
    % \TL{for all $\formule \in \{\optim, \pessim\}$ and $\property \in \{\bounded, \unbounded\}$.}
\end{proposition}

\subsubsection*{Berge's Maximum Theorem and semicontinuity of bilevel value functions} 
For the analysis in the bounded setting, we will use Berge's maximum theorem. Its presentation involves the notion of outer semicontinuity of compact set-valued maps.

\begin{definition}[Outer semicontinuity]
    A compact set-valued map $\Theta: \cX \rightrightarrows \cY$ is called \emph{outer semicontinuous} if, and only if, for all sequences $(\firstvar_k)_{k \in \NN}$ of $\cX$ and $(\secondvar_k)_{k \in \NN}$ such that $\secondvar_k \in \Theta(\firstvar_k)$, if $\lim_{k \to \infty} \firstvar_k = \firstvar, \lim_{k \to \infty} \secondvar_k = \secondvar$, then $\secondvar \in \Theta(\firstvar)$.
\end{definition}

In the following, we provide a simplified version of Berge's Maximum Theorem to keep our discussion as simple as possible. 
\begin{theorem}[Berge's Maximum Theorem {\cite[Section 6.3]{maximum_theorem}}]
    \label{theorem:maximum_theorem}
    Consider a continuous function $g: \RR^n \times \RR^m$ and a compact set $\cY$. Define $\Theta(\firstvar)=\argmin\{g(\firstvar,\secondvar): \secondvar \in \cY\}$, we have $\Theta: \RR^n \rightrightarrows \cY$ is an outer semicontinuous set-valued mapping with non-empty and compact values.
\end{theorem}

As a consequence of Berge's Theorem, one obtains classical semicontinuity properties for bilevel value functions for which we provide a proof for completeness. Recall that a function $f: \RR^n \to \RR$ is lower semicontinuous if, and only if, its epigraph $\text{epi}(f):=\{(x, \alpha) \in \RR^n \times \RR \mid f(x) \leq \alpha\}$ is closed. %These notions allow establishing an ``upper-bound'' for $\cP^{\optim}_{\bounded}, \cP^{\pessim}_{\bounded}$.
% \begin{corollary}[Semicontinuity of bilevel value functions]
%     \label{lemma:upper-bound-compact-nonconvex-class}
%      Let $P, Q$ and $\cY$ as in \eqref{eq:simple-bilevel-poly-optim} be such that $P$ and $Q$ are continuous and $\cY$ is compact. Then, \Cref{assumption:well-definedness} is satisfied, and the corresponding value function $\valoptim$ (resp. $\valpessim$) for the optimistic bilevel problem \eqref{eq:optimistic-bilevel} (resp. pessimistic bilevel problem \eqref{eq:pessimistic-bilevel}) is lower (resp. upper) semicontinuous and bounded on every compact.
% \end{corollary}

\begin{corollary}[Semicontinuity of bilevel value functions with compact lower-level feasible sets]
    \label{lemma:upper-bound-compact-nonconvex-class}
     Let $P, Q$ and $\cY$ as in \eqref{eq:simple-bilevel-poly-optim} be such that $P$ and $Q$ are continuous and $\cY$ is compact. Then, for all $\firstvar \in \RR^n$, $\Theta(\firstvar)$ is non-empty and $P(\firstvar, \cdot)$ attains its minimum in $\Theta(\firstvar)$. Moreover, the corresponding value function $\valbilevel$ is lower semicontinuous and bounded on every compact.
\end{corollary}
\begin{proof}
    \begin{enumerate}[leftmargin=*]
        %\item Semi-algebraicity of $\valoptim$: This is a corollary of \Cref{lemma:upper-bound-nonconvex-class}.
        \item Non-emptiness of $\Theta(\firstvar)$: if $\cY$ is compact, then the set $\Theta(\firstvar)$ is non-empty and compact thanks to \Cref{theorem:maximum_theorem}. Thus, the continuous $P$ attains its minimum in $\Theta(\firstvar)$.
        % \item Lower-semicontinuity of $\valoptim$: consider the function
        % \begin{equation*}
        %     H(\firstvar, \secondvar) = \iota_{\Theta(\firstvar)}(\secondvar):= \begin{cases}
        %         0, & \text{if } \secondvar \in \Theta(\firstvar)\\
        %         +\infty, & \text{otherwise}
        %     \end{cases}.
        % \end{equation*}
        % Since $\Theta$ is a compact set-valued mapping and outer semicontinuous (thanks to \Cref{theorem:maximum_theorem}), $H(\cdot,\secondvar)$ is lower semicontinuous. As a consequence, $K_\secondvar(\cdot) = P(\cdot, \secondvar) + H(\cdot, \secondvar)$ is also lower semicontinuous. Thus, for each fixed $\secondvar$, the epigraph of $K_\secondvar(\cdot)$ is closed.

        % Since $\valoptim(\firstvar) = \min_{\secondvar \in \Theta(\firstvar)} P(\firstvar, \secondvar) = \min_{\secondvar} K(\firstvar, \secondvar):= P(\firstvar, \secondvar) + H(\firstvar, \secondvar)$, the epigraph of $\valoptim = \cap_{\secondvar} K_\secondvar$, and thus closed (since the intersection of closed sets remains closed). Hence, $\valoptim$ is lower semicontinuous.
        \item Lower-semicontinuity of $\valbilevel$: Consider a point $\firstvar \in \RR^n$ and a sequence $(\firstvar_k)_{k \in \NN}$ converging to $\firstvar$. Since $\Theta(\firstvar)$ is compact, there exists $\secondvar_k \in \Theta(\firstvar_k)$ such that $\valbilevel(\firstvar_k) = P(\firstvar_k,\secondvar_k)$. We need to prove that: $\valbilevel(\firstvar) \leq \liminf_{k \to \infty} \valbilevel(\firstvar_k)$. 
        
        In the following, we can assume that $\valbilevel(\firstvar_k)$ converges to $\liminf_{k \to \infty} \valbilevel(\firstvar_k)$ and prove this limit is at least $\valbilevel(\firstvar)$. Due to the compactness of $\cY$, the sequence $\secondvar_k$ admits at least an accumulation point ${\secondvar}$. Due to the outer semicontinuity of $\Theta$, ${\secondvar} \in \Theta(\firstvar)$. Due to the continuity of $P$, we also have: $\liminf_{k \to \infty} \valbilevel(\firstvar_k) = P(x,y)$. Therefore,
        \begin{equation*}
            \valbilevel(\firstvar) = P(\firstvar,\secondvar^\star) \leq P(x,y) = \liminf_{k \to \infty} \valbilevel(\firstvar_k),
        \end{equation*}
        where $\secondvar^\star \in \argmin \{P(\firstvar,\secondvar) \mid \secondvar \in \Theta(\firstvar)\}$.
        \item Boundedness on compact sets of $\valbilevel$: Since $\valbilevel$ is lower-semicontinuous, it is lower-bounded in a given compact set $C$. In addition, $\forall \firstvar \in C$, we also have:
        \begin{equation*}
            \begin{aligned}
                \max_{\firstvar \in C} \valbilevel(\firstvar) &= \max_{\firstvar \in C} \min_{\secondvar \in \Theta(\firstvar)}P(\firstvar,\secondvar)
                \leq \max_{\firstvar \in C} \max_{\secondvar \in \cY} P(\firstvar,\secondvar) = \max_{(\firstvar,\secondvar) \in C \times \cY} P(\firstvar,\secondvar).
            \end{aligned}
        \end{equation*}
        Since $P$ is continuous and $C \times \cY$ is compact, there exists a constant $C$ such that $\max_{\firstvar \in C} \valbilevel(\firstvar) \leq C$. Thus, $\valbilevel$ is also upper-bounded in $C$. The proof is concluded. 
    \end{enumerate}   
\end{proof}
\begin{remark}
    As a consequence of \Cref{lemma:upper-bound-compact-nonconvex-class}, semi-algebraic functions that are not lower semicontinuous cannot be expressed as polynomial bilevel programs with compact lower-level feasible sets. %An example of such a function is:
   % \begin{equation*}
    %    f(x) = \begin{cases}
     %       -1 & \text{ if } x < 0 \\
      %      0 & \text{ if } x = 0 \\
       %     1 & \text{ if } x > 0
       % \end{cases},
   % \end{equation*}
   % since $\lim_{x \to 0^-} f(x) = -1 < f(0) < \lim_{x \to 0^+} f(x) = 1$. Thus, the bounded (or more precisely, compact) feasible sets in the lower-level problem help exclude certain ``pathological'' functions and hence deserves to be investigated specifically. 
\end{remark}

%\ed{I keep the following for the moment}\TL{I removed it}
% \begin{lemma}[Properties of continuous semi-algebraic bilevel optimization]
%     \label{lemma:upper-bound-compact-nonconvex-class}
%     Consider a bilevel formulation whose upper-level, lower-level functions $P, Q$ are continuous semi-algebraic and the lower-level feasible set $\cY$ is compact and semi-algebraic. The corresponding value function $\valoptim$ (resp. $\valpessim$) as in \eqref{eq:optimistic-bilevel} (resp. \eqref{eq:pessimistic-bilevel}) is:
%     \begin{enumerate}
%         \item Semi-algebraic.
%         \item Lower (resp. upper) semicontinuous.
%         \item Bounded in every compact set, i.e., for any compact set $S$, there exists a constant $C > 0$ such that $|f(x)| \leq C$. 
%     \end{enumerate}
% \end{lemma}

% \begin{proof}
  
% \end{proof}
{In the rest of this section, the main results are presented in the following plan:
\begin{enumerate}[leftmargin=*]
    \item \Cref{subsection:closed-results}: we study the class of value functions when $P,Q$ are polynomials and $\cY$ is a box.
    \item \Cref{subsection:compact-results}: we study the class of value functions when $P,Q$ are polynomials, $Q(x,\cdot)$ is convex and $\cY$ is either a box or a general convex set.
    \item \Cref{subsection:summary_results}: we summarize all the results in this section.
\end{enumerate}
}
\subsection{Value functions of general polynomial bilevel programs}
\label{subsection:closed-results}
% \jer{Make this section double or with two paragraphs\\
% Make clear what is $Y$\\
% Call a problem SBP with $P, Q$ and $Y$ simple semi-agebraic sets a {\em polynomial bilebvel problem}}\\
We introduce the following notation which is a shorthand for value function classes for bounded or unbounded polynomial bilevel problems.

\begin{definition}[Value functions classes]
    \label{def:class-general-value-functions}
    Given $\property \in \{\bounded, \unbounded\}$, we define $\cP_{\property}$ the set of value functions that can be represented as a bilevel optimization problem with $P,Q$ polynomials and $\cY$ a box satisfying $\property$ (bounded or unbounded), i.e.:
    % \begin{equation}
    %     \label{eq:nonconvex-class}
    %     \cP^{\formule}_{\property} := \{h: \RR^n \to \RR \mid \exists P, Q \text{ polynomials}, \cY \text{ $\property$ satisfy \Cref{assumption:well-definedness} and } \ff_{\formule} = h\}.
    % \end{equation}
    \begin{equation}
        \label{eq:nonconvex-class}
        \cP_{\property} := \{h: \RR^n \to \RR \mid \exists P, Q \text{ polynomials}, \cY \text{ satisfy $\property$ and } \valbilevel = h\}.
    \end{equation}
\end{definition}

For example, if one takes $\property = \unbounded$, we have:
% \begin{equation*}
%     \cP^{\optim}_{\unbounded} := \{h \mid \exists P, Q \text{ polynomials}, \cY \text{ unbounded satisfy \Cref{assumption:well-definedness} and } \valoptim = h\},
% \end{equation*}
\begin{equation*}
    \cP_{\unbounded} := \{h \mid \exists P, Q \text{ polynomials}, \cY \text{ unbounded and } \valbilevel = h\},
\end{equation*}
where $\valbilevel$ is defined in \eqref{eq:optimistic-bilevel}. 
%In the following, we will provide many sharp characterizations for the sets in \eqref{eq:nonconvex-class}. Our results reveal several intrinsic difficulties of bilevel optimization, even in rigid geometry (such as semi-algebraic geometry).
We consider the cases of unbounded and bounded $\cY$ separately.

With our notation, \Cref{cor:upper-bound-nonconvex-class} asserts that $\cP_{\unbounded} \subset \SA$. It is natural to ask whether these inclusions are tight and the following theorem provides a positive answer.
\begin{theorem}[Value functions of  polynomial bilevel programming]
    \label{theorem:closed=SA}
    Any extended-real-valued semi-algebraic function is the value function of a polynomial bilevel problem whose lower-level problem is unconstrained. In particular,
    \emph{$$\cP_{\unbounded} =  \SA.$$}
\end{theorem}
\begin{proof}[Sketch of proof for \Cref{theorem:closed=SA}]
    We provide a high-level idea of proof here:
    Given an extended-valued semi-algebraic function $h: \RR^n \to \RR \cup \{\pm \infty\}$, we partition $\RR^n$ into three disjoint semi-algebraic components: $\dom(h), \dom(h)^+, \dom(h)^-$ as in \Cref{def:extended-semi-algebraic-function}.
    
    The main idea of our construction is to build upper-level and lower-level functions $P, Q$ such that:
    \begin{enumerate}
        \item If $\firstvar \in \dom(h)$, then $\Theta(\firstvar) := \argmin_\secondvar Q(\firstvar, \cdot)$ is non-empty. Moreover, $\min_{\secondvar \in \Theta(\firstvar)} P(\firstvar, \secondvar) = h(\firstvar)$ (note that we use $\min$, instead of $\inf$, which implies that the minimum value is attained).
        \item If $\firstvar \in \dom(h)^+$, then $\Theta(\firstvar) = \emptyset$.
        \item If $\firstvar \in \dom(h)^-$, then $\Theta(\firstvar) \neq \emptyset$ but $\inf_{\secondvar \in \Theta(\firstvar)} P(\firstvar, \secondvar) = - \infty$.
    \end{enumerate}
\end{proof}
\begin{proof}[Proof of \Cref{theorem:closed=SA}]
     Consider an extended-real-valued semi-algebraic function $h\colon \RR^n \to \RR \cup \{ \pm \infty\}$. By definition, $\graph  h =\{(x,h(x))\}$ is a semi-algebraic set, and it can be written as:
    \begin{equation}
        \label{eq:graph-decomposition}
        \graph h = \bigcup_{i \in I} \cS_i \quad \text{where} \quad \cS_i := \{(x,t) \in\RR^{n+1} \mid P_i(x,t) = 0 \text{ and } Q_{ij}(x,t) > 0, j \in J\},
    \end{equation}
    where $I$ and $J$ are some finite index sets, the functions $P_i, i \in I$ and $Q_{ij}, (i,j) \in I \times J$ are polynomials w.r.t $x \in \RR^n$ and $t \in \RR$.

    Consider $\dom(h), \dom(h)^+, \dom(h)^-$ defined as in \Cref{def:extended-semi-algebraic-function}: WLOG, we also assume that:
    \begin{equation}
        \label{eq:domain-decomposition}
        \begin{aligned}
            \dom(h) &:= \bigcup_{i \in I} \cD_i \quad \text{where} \quad \cD_i := \{\firstvar \in \RR^n \mid R_i(\firstvar) = 0 \text{ and } S_{ij}(\firstvar) > 0\}\\
            \dom(h)^+ &:= \bigcup_{i \in I} \cD^+_i \quad \text{where} \quad \cD^+_i := \{\firstvar \in \RR^n \mid R^+_i(\firstvar) = 0 \text{ and } S^+_{ij}(\firstvar) > 0\}\\
            \dom(h)^- &:= \bigcup_{i \in I} \cD^-_i \quad \text{where} \quad \cD^-_i := \{\firstvar \in \RR^n \mid R^-_i(\firstvar) = 0 \text{ and } S^-_{ij}(\firstvar) > 0\}\\
        \end{aligned}
    \end{equation}
    with the same index sets $I,J$ as in \eqref{eq:graph-decomposition} (otherwise, one can add ``dummy'' polynomial equalities/inequalities, e.g., $0 = 0, 0 < 1$  to match the index sets). 
    
    In our construction, we use eight sets of variables that are described in \Cref{tab:variables-semi-algebraic-function}. In particular, the lower-level variable $\secondvar$ is the concatenation of $(t, z, \nu, \nu^+, \nu^-, u, v)$.

    \begin{table}[bhtp]
        \centering
        \begin{tabular}{cccc}
            \toprule
             Name & Dimension & Coordinates notation & Type (Upper/Lower variable) \\
             \midrule
             $x$ & $n$ & Not used & Upper \\ 
             \midrule
             $t$ & $1$ & Not used & Lower \\
             \midrule
             $z$ & $|I| \times |J|$ & $z_{ij}$ & Lower \\
             \midrule
             $\nu$ & $|I| \times |J| $ & $\nu_{ij}$ & Lower \\
             \midrule
             $\nu^+$ & $|I| \times |J| $ & $\nu_{ij}^+$ & Lower \\
             \midrule
             $\nu^-$ & $|I| \times |J| $ & $\nu_{ij}^-$ & Lower \\
             \midrule
             $u$ & $1$ & Not used & Lower \\
             \midrule
             $v$ & $1$ & Not used & Lower \\
             \bottomrule
        \end{tabular}
        \caption{Specification for the variables of the bilevel formulation.}
        \label{tab:variables-semi-algebraic-function}
    \end{table}

\noindent
    Consider the function:
    \begin{equation}
        \label{eq:def-F}
        \begin{aligned}
            F(\firstvar, t, z) :=  \prod_{i \in I} \left(P_i(\firstvar,t)^2 + \sum_{j \in J} G_{ij}(x,t,z_{ij})\right).
        \end{aligned}
    \end{equation}
    \begin{equation*}
        G_{ij}(x,t,z_{ij}) = (1 - Q_{ij}(x,t)z_{ij}^2)^2,
    \end{equation*}
    where $Q_{ij}$ are the polynomials defined in \Cref{eq:graph-decomposition}. Let us determine its global minimizers for a fixed value of $(x,t)$. One gets the closed form for the minimizers $z_{ij}^2$ and the optimal values $G_{ij}^\star$ of $G_{ij}(x,t,\cdot)$:
    \begin{equation}
        \label{eq:optimized-z-ij}
        z_{ij}^\star = \begin{cases}
            \pm \sqrt{\frac{1}{Q_{ij}(x,t)}} & \text{ if } Q_{ij}(x,t) > 0\\
            0 & \text{ otherwise,}
        \end{cases} \qquad \qquad
        G_{ij}^\star = \begin{cases}
            0 & \text{ if } Q_{ij}(x,t) > 0\\
            1 & \text{ otherwise}
        \end{cases}.
    \end{equation}
    Given a fixed value $x \in \dom(h)$, consider two cases:
    \begin{enumerate}[leftmargin=*]
        \item If $(x,t) \in \graph h$ (or equivalently, $t = h(x)$), there exists $i \in I$ such that $P_i(x,t) = 0$ and $Q_{ij}(x, t) > 0, \forall j \in J$. Using \eqref{eq:optimized-z-ij}, we have:
        \begin{equation*}
            P_i(x,t)^2 + \sum_{j \in J} G_{ij}(x,t,z_{ij}^\star) = 0.
        \end{equation*}
        Therefore, $F(x,h(x),z^\star) = 0$, in other words $z^\star$ is a global minimizer of $F(x,h(x),\cdot)$. 
        \item If $(x,t) \notin \graph h$ (or equivalently, $t \neq h(x)$), then for all $i \in I$, either $P_i(x,t) \neq 0$ or there is $j$ such that $Q_{ij}(x,t) \leq 0$. In any case, using \eqref{eq:optimized-z-ij} again, we can conclude that:
        \begin{equation*}
            P_i(x,t)^2 + \sum_{j \in J} G_{ij}(x,t,z_{ij}^\star) > 0, \forall i \in I.
        \end{equation*}
        Therefore, if $t \neq h(x)$, we have $F(x,t,z) > 0, \forall z \in \RR^{|I| \times |J|}$. 
    \end{enumerate}
    
    Analogous to \eqref{eq:def-F}, we construct three nearly similar polynomials, using the functions $S, R$ defined in \eqref{eq:domain-decomposition}.
    \begin{equation}
        \label{eq:def-H}
        \begin{aligned}
            H(\firstvar, \nu) :=  \prod_{i \in I} \left(R_i(\firstvar)^2 + \sum_{j \in J} K_{ij}(\firstvar,\nu_{ij})\right) \quad \text{where} \quad K_{ij}(\firstvar, \nu_{ij}) = (1 - S_{ij}(\firstvar)\nu_{ij}^2)^2\\  
            H^+(\firstvar, \nu) :=  \prod_{i \in I} \left(R^+_i(\firstvar)^2 + \sum_{j \in J} K^+_{ij}(\firstvar,\nu^+_{ij})\right) \quad \text{where} \quad K^+_{ij}(\firstvar, \nu^+_{ij}) = (1 - S_{ij}(\firstvar)[\nu^+_{ij}]^2)^2\\  
            H^-(\firstvar, \nu) :=  \prod_{i \in I} \left(R^-_i(\firstvar)^2 + \sum_{j \in J} K^-_{ij}(\firstvar,\nu^-_{ij})\right) \quad \text{where} \quad K^-_{ij}(\firstvar, \nu^-_{ij}) = (1 - S_{ij}(\firstvar)[\nu^-_{ij}]^2)^2\\  
        \end{aligned}
    \end{equation}
    Using a similar argument for $F$, we can conclude that: 
    \begin{enumerate}
        \item if $\firstvar \in \dom(h)$ (resp. $\dom(h)^+, \dom(h)^-$), there exists $\nu$ (resp. $\nu^+, \nu^-$) such that $H = 0$ (resp, $H^+ = 0, H^- = 0$);
        \item Otherwise, $H(\firstvar, \nu) > 0$ (resp. $H^+(\firstvar, \nu^+) > 0, H^-(\firstvar, \nu^-) > 0$) for all $\nu$ (resp. $\nu^+, \nu^-$).
    \end{enumerate}

    Using $F, H, H^+, H^-$, we construct the upper-level and lower-level polynomials $P$ and $Q$ as follows:
    \begin{equation*}
        \begin{aligned}
            P(\firstvar,\secondvar) &= t, \\
            Q(\firstvar,\secondvar) &= \underbrace{H^+(\firstvar, \nu^+)H^-(\firstvar, \nu^-)F(\firstvar, t, z)}_{Q_1(\firstvar, t, z, \nu^+, \nu^-)} + \underbrace{H(\firstvar, \nu)H^-(\firstvar, \nu^-)u^2 + (1 - uv)^2}_{Q_2(\firstvar, \nu, \nu^-, u, v)}.
        \end{aligned}
    \end{equation*}

    By construction, all functions $F, H, H^+, H^-$ are sums of squares. Therefore, given a fixed value $\firstvar$, if $Q(\firstvar, \secondvar) = 0$, then $\secondvar$ belongs to the set of minimizers of $Q(x, \cdot,\cdot)$. We consider three cases corresponding to the partition of $\RR^n$:
    \begin{enumerate}
        \item If $\firstvar \in \dom(h)$: on the one hand, $\min Q(\firstvar, \secondvar) = 0$ and it is attained since we can choose $(t,z, \nu, u, v)$ such that $F(\firstvar, t,z) = 0$ (see \cref{eq:optimized-z-ij}), $H(\firstvar,\nu) = 0$ (see remark after \cref{eq:def-H}), $uv = 1$. On the other hand, if $\secondvar$ minimizes $Q(\firstvar,\cdot)$, that is $\secondvar \in \Theta(\firstvar)$, then $t = h(\firstvar)$ due to our analysis of $F(\firstvar, t, z)$ and the fact that $H^+(\firstvar, \nu^+)H^-(\firstvar, \nu^-) > 0$, for any $\nu^+,\nu^-$. Thus, $P(\firstvar, \secondvar) = t = h(\firstvar)$. 
        \item If $\firstvar \in \dom(h)^+$: the infimum $\inf Q(\firstvar, \secondvar) = 0$ but it is not attained. Indeed, by choosing $\nu^+, u, v$ such that $H^+(\firstvar, \nu^+) = 0$ (see remark after \cref{eq:def-H}), $u \to 0$ and $v = 1/u \to +\infty$, $Q(\firstvar, \secondvar)$ can get arbitrarily close to zero. Nevertheless, the minimum is not attained since for any $\nu, \nu^-$, $H(\firstvar, \nu)H^-(\firstvar, \nu^-) > 0$ and one can verify that $au^2 + (1 - uv)^2 > 0$, for any $u,v$ for any $a > 0$, which concludes the proof for this case.
        \item If $\firstvar \in \dom(h)^-$: $\min Q(\firstvar, \secondvar) = 0$ and it is attained for any $t \in \RR$, by choosing $\nu^-, u, v$ such that $H^-(\firstvar,\nu^-) = 0$ (see remark after \cref{eq:def-H}) and $uv = 1$. Thus, the bilevel optimization results in $\inf P(\firstvar, \secondvar) = -\infty$. 
    \end{enumerate}
    That concludes the proof.
\end{proof}
{
\begin{remark}
    We remark that the degree of the constructed polynomial $Q$ is linear in the degrees of polynomials $P_i,Q_{ij}$ defining the graph of the target function $h$. Since all the information about a bilevel problem is encoded in a single pair of polynomials $(P,Q)$ it is natural that their degree increases depending on the complexity of the underlying representation. A similar comment holds for all the constructions of this section. We leave more quantitative discussions about this representation for future work.
\end{remark}
}
{
\begin{remark}
    While we focus on bilevel programming, the above results actually characterize semi-algebraic functions as a polynomial $\argmin$ since the upper level is just the projection on the first coordinate. We also remark that the proof allows to obtain a representation of semi-algebraic sets using a non-negative polynomial argmin or equivalently the zero locus of a polynomial.
\end{remark}
}
%When the lower-level feasible set $\cY$ is closed, there is not any difference between optimistic (cf. \eqref{eq:optimistic-bilevel}) and pessimistic (cf. \eqref{eq:pessimistic-bilevel}) formulation.

%Similar to the unbounded cases, we provide the results for its bounded counterpart. \TL{Unlike the bounded cases, we have: $\SA \neq \cP^{\optim}_{\bounded}$ and $\SA \neq \cP^{\pessim}_{\bounded}$, as seen in the following remark.} %Our first result in this section shows that if the lower-level feasible set is bounded instead of unbounded, then the sets $\cP^{\optim}_{\bounded}, \cP^{\pessim}_{\bounded}$ are strictly contained in $\SA$. 

Denote respectively by $\LSC, \BC$ the sets of functions that are lower semicontinuous, and bounded on any compact set (cf. \Cref{lemma:upper-bound-compact-nonconvex-class}). Similar to the unbounded case, combining \Cref{cor:upper-bound-nonconvex-class} and \Cref{lemma:upper-bound-compact-nonconvex-class}, we have that $\cP_{\bounded} \subset \SA\cap \LSC \cap \BC$. The following shows that this inclusion is tight. Note that, by definition, functions in $\BC$ have full domain (they do not take value $\pm \infty$).

%we will show in the following that $\cP^{\optim}_{\bounded}$ and $\cP^{\pessim}_{\bounded}$ can be characterized precisely using these sets.

\begin{theorem}[Value functions of  box-constrained polynomial bilevel programming]
    \label{theorem:compact=lsc+bc}
    Any function which is semi-algebraic, lower semicontinuous, and bounded on compact sets is the value function of a polynomial bilevel problem whose lower-level feasible set $\cY$ is a bounded box. In other words:
    \begin{equation*}
        \cP_{\textup{\bounded}} = \SA \cap \LSC \cap \BC.
    \end{equation*}
\end{theorem}
\begin{proof}
    To prove the equalities in \Cref{theorem:compact-convex=LSC}, we notice that \Cref{cor:upper-bound-nonconvex-class} and \Cref{lemma:upper-bound-compact-nonconvex-class} imply the following:
    \begin{equation*}
        \begin{aligned}
            \cP_{\bounded} &\subseteq \SA \cap \LSC \cap \BC,\\
        \end{aligned}
    \end{equation*}
    Therefore, it is sufficient to prove the first claim of \Cref{theorem:compact-convex=LSC}: for any element $f: \RR^n \to \RR$ of $\SA \cap \LSC \cap \BC$, there exists a polynomial bilevel problem with a bounded box $\cY$ whose value function equals $f$. Note that if $f$ is bounded on any compact, $f(\firstvar) \in \RR, \forall \firstvar \in \RR^n$. Before constructing $P, Q$ and the bounded set $\cY$, we state three key observations concerning $f$.
    \begin{enumerate}[leftmargin=*]
        \item Since $f$ is semi-algebraic, $\graph f$ is also semi-algebraic. Consequently, the closure $\overline{\graph f}$ is also semi-algebraic (see, for example, \cite[Proposition 3.1]{realgeoalg}). By \Cref{lemma:closed-semi-algebraic-set}, $\overline{\graph f}$ can be represented as:
        \begin{equation}
            \label{eq:proof-closed-semi-algebraic-set}
            \overline{\graph f} = \bigcup_{i \in I}\bigcap_{j \in J} \{(x, t) \in \RR^{n+1} \mid \, P_{ij}(x,t) \geq 0\},
        \end{equation}
        where $P_{ij}$ are polynomials with $n+1$ variables, and $I,J$ are two finite index sets. 
        \item Since $f$ is lower-semicontinuous one has  $f(x) = \min \{t \mid t \in \RR, (x,t) \in \overline{\graph f}\}$.
        \item Since $f$ is semi-algebraic, by \Cref{lemma:growth-semi-algebraic-functions}, there exist a positive constant $C$ and an integer $N$ such that $|f(x)| \leq \|x\|^N, \forall x, \|x\|\geq C$. We can choose $N$ to be an even integer so that $\|x\|^N$ is a polynomial. On the other hand, since $f$ is bounded on compact sets, there exists another positive constant $B$ such that $|f(x)| \leq B, \forall x, \|x\| \leq C$. By combining these two observations, we have:
        \begin{equation*}
            f(x) \leq B + \|x\|^N, \forall x \in \RR^n.
        \end{equation*}
       Using this remark, we may assume that $f$ is bounded in $[-1,1]$. 
       
       Indeed, consider the function: $h(x) = f(x) / (B + \|x\|^N)$. Since the polynomial $B + \|x\|^N \geq B, \forall x \in \RR$, $h(x)$ is well-defined, and it remains semi-algebraic, lower-semicontinuous. Moreover, $h$ is bounded in $[-1,1]$. If one can construct $P,Q$ and $\cY$ such that $\valbilevel$ in \eqref{eq:optimistic-bilevel} satisfies $\valbilevel = h$, then the bilevel problem with $(P(B + \|x\|^N), Q, \cY)$ has a value function equal to $f$.
    \end{enumerate}
    In our construction, we use three sets of variables. They are described as in \Cref{tab:variables-lower-semicontinuous-function}.
    \begin{table}[bhtp]
        \centering
        \begin{tabular}{ccccc}
            \toprule
             Name & Dimension & Coordinates notation & Type (Upper/Lower variable) & Feasible set \\
             \midrule
             $x$ & $n$ & Not used & Upper& N/A\\ 
             \midrule
             $t$ & $1$ & Not used & Lower & $[-1,1]$\\
             \midrule
             $z$ & $|I| \times |J|$ & $z_{ij}$ & Lower & $[0,1/2]^{|I|\times|J|}$ \\
             \bottomrule
        \end{tabular}
        \caption{Specification for the variables of the bilevel formulation.}
        \label{tab:variables-lower-semicontinuous-function}
    \end{table}
    
    We introduce the building block of our polynomials $P$ and $Q$. Consider:
    \begin{equation*}
        G_{ij}(x,t,z_{ij}) = \left(P_{ij}(x,t) - (P_{ij}(x,t)^2 + 1)z_{ij}\right)^2,
    \end{equation*}
    where $P_{ij}$ are the polynomials defined in \eqref{eq:proof-closed-semi-algebraic-set}.
    
    Given a fixed value of $(x,t)$, optimizing $G_{ij}(x,t,z_{ij})$ w.r.t $z_{ij} \in [0, 1/2]$, the optimal value and minimizer of $G_{ij}$ are given by:
    \begin{equation}
        \label{eq:proof-semi-optimal-values}
        z_{ij}^\star = \begin{cases}
            \frac{P_{ij}(x,t)}{P_{ij}(x,t)^2 + 1} & \text{ if } P_{ij}(x,t) \geq 0\\
            0 & \text{ otherwise}
        \end{cases} \qquad, \qquad
        G_{ij}^\star(x,t) = \begin{cases}
            0 & \text{ if } P_{ij}(x,t) \geq 0\\
            P_{ij}(x,t)^2 & \text{otherwise}
        \end{cases}
    \end{equation}
    Note indeed that $z^\star_{ij} \in [0, 1/2]$ because $t/(1 + t^2) \in [0,1/2], \forall t \geq 0$. As a consequence $G^\star_{ij}(x,t) = 0$ if and only if $P_{ij}(x,t) \geq 0$. 

    Let us now define $P, Q$ and $\cY$ as:
    \begin{equation}
        \label{eq:construction-compact-lsc}
        \begin{aligned}
            P(x,t,z) &= t, \\
            Q(x,t,z) &= \prod_{i \in I}\left(\sum_{j \in J} G_{ij}(x,t,z_{ij})\right),\\
            \cY &= [-1, 1] \times [0, 1/2]^{|I| \times |J|}.
        \end{aligned}
    \end{equation}
    Note that $Q$ is the product of sums of squares. Hence, $Q(x,t,z) \geq 0$. Consequently, if $Q(x,t,z) = 0$, then $(t,z) \in \argmin_{\cY} Q(x, \cdot,\cdot)$. 

    Consider a point $x \in \RR^n$, there are two possibilities:
    \begin{enumerate}[leftmargin=*]
        \item If $(x,t) \in \overline{\graph f}$, then there exists $i \in I$ such that $P_{ij}(x,t) \geq 0, \forall j \in J$. By \eqref{eq:proof-semi-optimal-values}, we have:
        \begin{equation*}
            \sum_{j \in J} G_{ij}^\star(x,t,z_{ij}) = 0.
        \end{equation*}
        Thus, the optimal value of $Q$ in this case is zero. 
        
        \item If $(x,t) \notin \overline{\graph f}$, then for all $i \in I$, there exists at least an index $j_i \in J$ such that $P_{ij}(x,t) < 0$. Hence, 
        \begin{equation*}
            \sum_{j \in J} G_{ij}(x,t,z_{ij}^\star) \geq P_{ij_i}(x,t)^2 > 0, \forall i \in I.
        \end{equation*}
        Thus, the optimal value of $Q$ is at least $\prod_{ij} P_{ij_i}(x,t)^2 > 0$. 
    \end{enumerate}
    Therefore, for a minimizer $(t,z)$ of $Q(x, \cdot,\cdot)$, $(x,t) \in \overline{\graph f}$. We emphasize that such $t$ always exists and belongs to the interval $[-1,1]$ due to the hypothesis of boundedness of $f$. Finally, among $\{t \in [-1,1] \mid (x,t) \in \overline{\graph f}\}$, the optimistic formulation will choose the smallest $t$, which is exactly the value of $f(x)$ due to the lower-semicontinuity of $f$. 
\end{proof}

\subsection{Polynomial bilevel problems with convex lower-level}
\label{subsection:compact-results}
We also investigate the role of convexity of the lower-level problem in the set of expressible functions. Analogously to \Cref{def:class-general-value-functions}, under the assumption that the lower-level problem is convex, we study the following value function class, where the letter $\cC$ highlights convexity.

\begin{definition}[Value functions with lower-level convexity]
    \label{def:class-lower-level-convex}
    Under the same notations as in \Cref{def:class-general-value-functions}, the sets $\cC_{\property}$ are analogously defined to $\cP_{\property}$, except that the lower-level problem is constrained to be convex in the definition of the former, i.e.:
    % \begin{equation}
    %     \label{eq:convex-class}
    %     \cC^{\formule}_{\property} := \left\{h: \RR^n \to \RR \mid \begin{array}{c} \exists P, Q \text{ polynomials}, Q \text{ convex w.r.t } y, \cY \text{ $\property$ satisfy} \\\text{\Cref{assumption:well-definedness} such that } \ff_\formule = h\end{array} \right\}.
    % \end{equation}
    \begin{equation}
        \label{eq:convex-class}
        \cC_{\property} := \left\{h: \RR^n \to \RR \mid \begin{array}{c} \exists P, Q \text{ polynomials}, Q \text{ convex w.r.t } y, \cY \text{ $\property$} \text{ such that } \valbilevel = h\end{array} \right\}.
    \end{equation}
\end{definition}
Similar to the previous section, we treat the case of bounded and unbounded lower-level constraints separately. This section is concluded with a third result for which we allow the lower-level constraint set to be an arbitrary convex compact semi-algebraic set.

From \Cref{def:class-general-value-functions} and \Cref{def:class-lower-level-convex}, we clearly have that: 
\begin{equation*}
    \cC_{\property} \subseteq \cP_{\property},
\end{equation*}
for $\property \in \{\unbounded,\bounded\}$. 
However, we do not know if this inclusion is strict. Nevertheless, in this section, we show that the class of value functions in $\cC_{\unbounded}$, is very large as it contains all \emph{piecewise polynomial functions}. We denote by $\mbf{1}_S$ be the characteristic function of a subset $S \subseteq \RR^n$ (with value $1$ on $S$ and $0$ elsewhere), a piecewise polynomial can be defined as follows:

\begin{definition}[Piecewise polynomial functions]
    \label{def:piecewise-polynomial}
     A function $f: \RR^n \to \RR$ is called \emph{piece-wise polynomial} if there exist a semi-algebraic partition $S:= \{S_i, i = 1, \ldots, N\}$ of $\RR^n$ (i.e., $S_i$ are semi-algebraic, pairwise disjoint and their union is equal to $\RR^n$) and $N$ polynomials $P_i: \RR^n \to \RR$ such that:
    \begin{equation*}
        f(x) = \sum_{i=1}^N \mbf{1}_{x \in S_i}P_i(x), \forall x \in \RR^n.
    \end{equation*}
    Equivalently, $f(x) = P_i(x)$ for $x \in S_i$. We call $(S, \tupfuncN)$ the representation of $f$ and write $f = (S, \tupfuncN)$ by an abuse of notations. We use the shorthand $\PP$ to denote the set of piecewise polynomial functions.
\end{definition}

% Piecewise polynomial functions are related to the elementary (log-exp) functions \cite[Definition 3]{bolte2020mathematicalmodel}, but they are more restricted since their partitions are semi-algebraic (thus, do not contain log and exponential functions). However, as shown in \cite[Section 5]{bolte2021implicitdifferentiation}, piecewise polynomial functions are already problematic enough to make the traditional gradient descent method generate cyclic iterates and/or have Lorenz chaotic behavior \cite{Lorenz2004}. 

In the following, we show that the set of piecewise polynomial functions is contained in the set of value functions with convex lower-level and simple box constraints. This result illustrates that relaxing strong convexity but maintaining convex lower levels in bilevel programs allows to represent the large class of piecewise polynomial functions. While this is a strictly smaller class compared to semi-algebraic functions, this is still a very large class which contains functions which are generally discontinuous with an arbitrary number of discontinuities, and gradient type methods are not adapted to such functions \cite{bolte2021implicitdifferentiation}. 

\begin{theorem}[Piecewise polynomials are value functions with convex lower-level]
    \label{theorem:convex-closed-bilevel-class}
    Any piece-wise polynomial function is the value function of a polynomial bilevel problem with lower objective convex over a box. More specifically:
    \begin{equation*}
        \begin{aligned}
            \PP \; \subseteq \; \cC_{\textup{\unbounded}} \; \subseteq \; \SA.
        \end{aligned}
    \end{equation*}
\end{theorem}

% \begin{theorem}
% \label{theorem:convex-closed-bilevel-class}
%     For any piecewise polynomial function $f$, there exists two polynomials $P,Q$ and a closed lower-level feasible set $\cY$ such that:
%     \begin{enumerate}
%         \item $Q$ and the set $\cY$ are convex, i.e., the lower-level optimization problem is convex.
%         \item The corresponding optimistic (resp. pessimistic) bilevel problem equals $f$.
%     \end{enumerate}
% \end{theorem}

The proof of \Cref{theorem:convex-closed-bilevel-class} is based on the following lemma.
\begin{lemma}[Bilevel formulation for semi-algebraic characteristic functions]
    \label{lem:constrained-convex-construction}
    Consider a basic semi-algebraic set $S$ (cf. \Cref{def:semi-algebraic-set}). There exists a polynomial bilevel problem with a convex lower-level over an unbounded box whose value function is the characteristic function $\mbf{1}_S$.
\end{lemma}
\begin{proof}
    Let $S$ be of the form:
    \begin{equation}
        \label{eq:basic-elementary-function-proof}
        S:= \{P(x) = 0 \text{ and } Q_j(x) > 0, j \in J\} \subseteq \RR^n,
    \end{equation}
    where $P$ and $Q_j, j \in J$ are polynomials. Consider the following bilevel problem:
    \begin{equation*}
        \begin{aligned}
        &\min_{x} && F(x, w, z, t):= (1 - tP(x))\prod_{j = 1}^{|J|} \left(w_jz_jQ_j(x)\right)\\
        &\text{s.t.} && (w,z,t) \in \underset{w \in [0,1]^{|J|}, z \in \RR^{|J|}, t \in \RR}{\argmin} G(x,w,z,t) := (1 - tP(x))^2 + \sum_{j \in J} (1 - Q_j(x)z_j)^2 - Q_j(x)w_j
    \end{aligned} 
    \end{equation*}
    where $w_j, z_j$ indicate respectively the $j$th coordinate of the vectors $w$ and $z$ in $\RR^{|J|}$. 
    
    By construction, $F, G$ are polynomials. In addition, the lower-level problem is also clearly convex because given a fixed $x$, $G(x,\cdot,\cdot,\cdot)$ is linear w.r.t to $w$ and (semi-definite positive) quadratic w.r.t $z$ and $t$. 
    It remains to prove that the constructed bilevel problem has the value function equal to $\mbf{1}_S$. As we will see, given a fixed $x$, although there are multiple minimizers $(w,z,t)$, they all yield the same value $F(x,w,z,t)$.
    
    Due to the convexity and -- more importantly -- the separation of variables $w,z,t$, we can specify the optimal solution $(w^\star(x),z^\star(x),t^\star(x))$ of $G(x,\cdot,\cdot,\cdot)$ for each fixed $x \in \RR^n$ as follows:
    \begin{equation*}
        w_j^\star(x) = \begin{cases}
            1 & \text{ if } Q_j(x) > 0\\
            [0,1] & \text{ if } Q_j(x) = 0\\
            0 & \text{ if } Q_j(x) < 0
        \end{cases}, j \in J,
    \end{equation*}
    \begin{equation*}
        z_j^\star(x) = \begin{cases}
            1/Q_j(x) & \text{ if } Q_j(x) \neq 0\\
            \RR & \text { otherwise}
        \end{cases}, j \in J, \qquad\qquad t^\star(x) = \begin{cases}
            1/P(x) & \text{ if } P(x) \neq 0\\
            \RR & \text { otherwise}
        \end{cases}.
    \end{equation*}
    Therefore,
    \begin{equation*}
        w_j^\star(x)z_j^\star(x)Q_j(x) = \begin{cases}
            1 & \text{ if } Q_j(x) > 0 \\
            0 & \text{ otherwise}
        \end{cases}.
    \end{equation*}
    Similarly, we have:
    \begin{equation*}
        1 - t^\star(x)P(x) = \begin{cases}
            1 & \text{ if } P(x) = 0\\
            0 & \text{ otherwise}
        \end{cases}.
    \end{equation*}
    The result follows immediately from the two above equalities.
\end{proof}

\begin{proof}[Proof of \Cref{theorem:convex-closed-bilevel-class}]
    Consider the piecewise polynomial $f = (S, \tupfuncN)$ as in \Cref{def:piecewise-polynomial}.
    Due to the definition of semi-algebraic sets (cf. \Cref{def:semi-algebraic-set}), for all $i = 1, \ldots, N$, $S_i$ can be written as:
    \begin{equation*}
        S_i = \bigcup_{j \in J}^M T_{ij},
    \end{equation*}
    where $T_{ij}, i = 1, \ldots, N, j \in J$ are basic pairwise disjoint semi-algebraic -- use \Cref{lemma:stratification-elementary-set} in \Cref{appendix:alternative-def} and the fact that the $S_i, 1 \leq i \leq N$ are disjoint.
    For each $T_{ij}$, we take two polynomials $F_{ij}$, $G_{ij}$ (where $G_{ij}$ is convex w.r.t $\secondvar$) and an unbounded box $\cY_{ij}$ such that the following bilevel problem:
    \begin{equation*}
        \begin{aligned}
            &\min_{\firstvar} && F_{ij}(\firstvar, \secondvar(\firstvar))\\
            &\text{s.t.} && \secondvar(\firstvar) \in \argmin_{\secondvar\in\cY_{ij}} G_{ij}(\firstvar, \secondvar).
    \end{aligned} 
    \end{equation*}
    has value function equal to the characteristic function $\mbf{1}_{T_{ij}}$. Such polynomials and sets exist by the representability result for characteristic functions \Cref{lem:constrained-convex-construction}. 

    Consider the following polynomial bilevel problem:
    \begin{equation*}
        \begin{aligned}
            &\min_{\firstvar} && P(\firstvar,\secondvar): = \sum_{i = 1}^N P_i(\firstvar)\left(\sum_{j \in J} F_{ij}(\firstvar, \secondvar_{ij}(\firstvar))\right)\\
            &\text{s.t.} && \secondvar(\firstvar) := (\secondvar_{ij}(x))_{1 \leq i \leq {N}, j \in J} \in \underset{\secondvar_{ij}\in\cY_{ij}}{\argmin} \; Q(\firstvar, \secondvar) := \sum_{i,j} G_{ij}(\firstvar, \secondvar_{ij}).
    \end{aligned} 
    \end{equation*}
    Due to the separation of variables $y_{ij}$, we have:
    \begin{equation*}
        \underset{\secondvar_{ij}\in\cY_{ij}}{\argmin} \sum_{i,j} G_{ij}(\firstvar, \secondvar_{ij}) = \bigotimes_{i, j} \argmin_{y_{ij}} \; G_{ij}(x, y_{ij}).
    \end{equation*}
    Therefore, by \Cref{lem:constrained-convex-construction}, we have for all $i \in I$:
    \begin{equation*}
        \sum_{j \in J} F_{ij}(\firstvar, \secondvar_{ij}(\firstvar)) = \sum_{j \in J} \mbf{1}_{T_{ij}} = \mbf{1}_{S_i},
    \end{equation*}
    where the second equality holds because the $S_{ij}$ are pairwise disjoint. Therefore, $$P(x,y(x)) = \sum_{i = 1}^N \mbf{1}_{S_i}(x)P_i(x) = f(x)$$ as we desire. In addition, $Q(\firstvar,\secondvar)$ is convex since it is equal to the sum of convex functions. This concludes the proof.
\end{proof}

Analogous to \Cref{theorem:convex-closed-bilevel-class}, we provide a class of functions that can be expressed using polynomial bilevel problems with convex lower-level problems and bounded box $\cY$. Similarly, the boundedness assumption on the lower-level imposes a semicontinuity restriction on the underlying value function class. Note that piecewise polynomial functions are bounded on compact sets by construction.

\begin{theorem}[Semicontinuous piecewise polynomials are value functions with box constraints and convex lower-level]
\label{theorem:convex-compact-bilevel-class}
    Any lower semicontinuous piece-wise polynomial function is the value function of a polynomial bilevel problem with lower-level objective convex over a bounded box, i.e.
    \begin{equation*}
        \begin{aligned}
            \PP \cap \LSC \subseteq \cC_{\textup{\bounded}} \subseteq \SA \cap \LSC \cap \BC.
        \end{aligned}
    \end{equation*}
\end{theorem}

\begin{proof}
    In the following, we show that if $f \in \PP \cap \LSC$, then $f \in \cC_{\bounded}$, which is the first inclusion. The second inclusion was already justified in \Cref{theorem:compact=lsc+bc}.
    
    Consider $f = (S, \tupfuncN)$ a lower semicontinuous piecewise polynomial function. Our construction is based on two observations:
    \begin{enumerate}
        \item Given a point $x \in \RR^n$, define $\cI(x) := \{i \in \cI \mid x \in \overline{S_i}\}$ the subset of indices in which $x$ belongs to the closure of corresponding sets. Since $f$ is lower-semicontinuous, we have:
        \begin{equation}
            \label{eq:minimum-selection}
            f(x) = \min_{i \in \cI(x)} P_i(x).
        \end{equation}
        \item Similar to the proof of \Cref{theorem:compact=lsc+bc}, there exists a pair integer $N$ and a positive constant $B > 0$ such that $\max_{i = 1, \ldots, N} P_i(x) \leq B + \|x\|^N$. Thus, in the following construction, WLOG, one can assume that $P_i(x) \leq 0, \forall i = 1, \ldots, N, \forall x \in \RR^n$. Otherwise, we can consider the function $r(x) = f(x) - B - \|x\|^N$, which is also an element of $\cS$, lower-semicontinuous and all $r_i(x) = P_i(x) - B - \|x\|^N \leq 0$ for all $x$. If we can find functions $P, Q$ and a bounded box so that the corresponding bilevel problem has the value function equal to $r(x)$, then that of $(P + B + \|x\|^N, Q, \cY)$ is equal to $h$. 
    \end{enumerate}
    For each $i=1,\ldots, N$, since $S_i$ is semi-algebraic, so is its closure. By \Cref{lemma:closed-semi-algebraic-set}, they can be represented by:
    \begin{equation}
        \label{eq:closure-elementary-set-representation}
        \bar{S_i} = \bigcup_{j \in J}\bigcap_{k \in K} \{x \mid P^i_{jk}(x) \geq 0\}.
    \end{equation}
    for some index sets $J, K$ and polynomials $P^i_{jk}$. These polynomials will appear in our construction. 

    Similar to the proofs of other results, we introduce the sets of variables that will be used in our construction. 
    \begin{table}[bhtp]
        \centering
        \begin{tabular}{ccccc}
            \toprule
             Name & Dimension & Coordinates notation & Type (Upper/Lower variable) & Feasible set \\
             \midrule
             $x$ & $n$ & Not used & Upper& N/A\\ 
             \midrule
             $z$ & $|I||J||K|$ & $z_{jk}^i$ & Lower & $[0,1]^{|I| \times |J| \times |K|}$ \\
             \midrule
             $s$ & $|I||J|$ & $s_{ij}$& Lower & $[0,1]^{|I|\times|J|}$\\
             \bottomrule
        \end{tabular}
        \caption{Specification for the variables of the bilevel formulation.}
        \label{tab:variables-lower-compact-semicontinuous-function}
    \end{table}
    
    Denote $z^i_j = (z^i_{jk})_{k \in K}$, consider the optimization problem:
    \begin{equation*}
        \begin{aligned}
            \underset{z^i_j \in [0,1]^{|K|}}{\text{Minimize}} \quad && G^i_j(x,z^i_j) & = - \sum_{k \in K} P^i_{jk}(x)z^i_{jk},\\
        \end{aligned}
    \end{equation*}
    where $P_{jk}^i$ are polynomials defined in \eqref{eq:closure-elementary-set-representation}. Given a fixed value of $x$, the optimal value of $(z^i)^\star$ is given by:
    \begin{equation*}
        (z_{jk}^i)^\star = \begin{cases}
            0 & \text{ if } P_{jk}^i(x) < 0\\
            [0,1] & \text{ if } P_{jk}^i(x) = 0\\
            1 & \text{ if } P_{jk}^i(x) > 0
        \end{cases}.
    \end{equation*}
    Therefore, we can conclude that:
    \begin{equation}
        \label{eq:optimized-product-z-i}
        \prod_{k \in K} (z_{jk}^i)^\star = \begin{cases}
            0 & \text{ if } x \notin \bigcap_{k \in K} \{P^i_{jk}(x) \geq 0\}\\
            [0,1] & \text{ if } x \in \bigcap_{k \in K} \{P^i_{jk}(x) \geq 0\} \text{ and } \exists k \in K, P_{jk}^i(x) = 0\\
            1 & \text{ if } x \in \bigcap_{k \in K} \{P^i_{jk}(x) \geq 0\} \text{ and } \forall k \in K, P_{jk}^i(x) > 0\\
        \end{cases}.
    \end{equation}
    We consider the following bilevel problem:
    \begin{equation}
        \label{eq:bilevel-construction-convex-bounded}
        \begin{aligned}
            P &= \sum_{(i,j) \in I \times J} P_i(x)\left(s_{ij}\prod_{k \in K} z_{jk}^i\right),\\
            Q &= (1 - \sum_{(i,j) \in I \times J} s_{ij})^2 + \sum_{(i,j) \in I \times J} G^i_j(x,z^i_j),\\
            s_{ij} & \in [0,1], \forall (i,j) \in I \times J,\\
            z_{jk}^i & \in [0,1], \forall (i,j,k) \in I \times J \times K.  
        \end{aligned}
    \end{equation}
    The lower-level problem is obviously convex since $(1 - \sum_{(i,j) \in I \times J} s_{ij})^2$ is convex w.r.t $s_{ij}$, $G^i_j$ is linear w.r.t $z_j^i$ and the feasible set is a hypercube.
    
    In particular, $P_i(x)s_{ij} \leq 0$ since $s_{ij} \in [0,1]$ and $P_i(x)$ is assumed to be negative. Due to the optimistic nature of the bilevel problem and \eqref{eq:optimized-product-z-i}, we get:
    \begin{equation}
        \label{eq:optimize-product-h-z}
        P_i(x)\prod_{k \in K} z_{jk}^i = \begin{cases}
            0 & \text{ if } x \notin \bigcap_{k \in K} \{P^i_{jk}(x) \geq 0\}\\
            P_i(x) & \text{ if } x \in \bigcap_{k \in K} \{P^i_{jk}(x) \geq 0\}\\
        \end{cases}.
    \end{equation}
    Moreover, any $(s_{ij})_{(i,j) \in |I| \times |J|} \in [0,1]^{|I| \times |J|}$ such that $\sum_{i,j} s_{ij} = 1$ is optimal for $Q$. Since we consider the optimistic bilevel formulation, the value of $P$ in  \eqref{eq:bilevel-construction-convex-bounded} becomes:
    \begin{equation*}
        \begin{aligned}
            &\underset{s_{i,j} \geq 0, \sum s_{ij} = 1}{\min} \sum_{(i,j) \in I \times J} s_{ij} \left(P_i(x)\prod_{k \in K} z^i_{jk} \right).\\
            %\overset{\eqref{eq:optimize-product-h-z}}{=}& \underset{s_{i,j} \geq 0, \sum s_{ij} = 1}{\text{Minimize}} \sum_{(i,j) \in I \times J} s_{ij} \left(h_i(x)\mbf{1}_{x \in \bigcap_{k \in K} \{P^i_{jk}(x) \geq 0\}}\right)\\
        \end{aligned}
    \end{equation*}
    Thus, the value of $P$ will be equal to the smallest value of $P_i(x)\prod_{k \in K} z^i_{jk}$. This value will be equal to $f(x)$ because:
    \begin{equation*}
        \begin{aligned}
            \underset{(i,j) \in I \times J}{\min} \left\{P_i(x)\prod_{k \in K} z^i_{jk}\right\} &\overset{\eqref{eq:optimize-product-h-z}}{=} \underset{(i,j) \in I \times J}{\min} \left\{P_i(x)\mbf{1}_{x \in \bigcap_{k \in K} \{P^i_{jk}(x) \geq 0\}}\right\} \\
            &\overset{\eqref{eq:closure-elementary-set-representation}}{=} \underset{i \in I}{\min} \left\{P_i(x)\mbf{1}_{x \in \overline{s_i}}\right\} \\
            &\overset{\eqref{eq:minimum-selection}}{=} f(x).
        \end{aligned}
    \end{equation*}
    This concludes the proof. 
\end{proof}

\subsubsection{Extension to arbitrary convex, compact, semi-algebraic lower-level constraints}

\Cref{theorem:convex-closed-bilevel-class} and \Cref{theorem:convex-compact-bilevel-class} are limited to box lower-level constraint set $\cY$. While these constraints are explicit, this leaves open the question of the tightness of the corresponding inclusions. We will not answer this precise question here, but will consider a related question by allowing the lower-level constraint set $\cY$ to be an arbitrary convex compact semi-algebraic set. We remark that for such a lower-level constraint set $\cY$, the leftmost inclusion in \Cref{theorem:convex-compact-bilevel-class} is strict. For example, the Euclidean norm can be expressed as a maximum over a ball and is not piecewise polynomial.

We will use the following shorthand to describe the corresponding class of value functions which admit the required bilevel representation:
\begin{equation}
    \label{eq:convex-compact-class}
    \cC_{\coco} := \{h \mid \exists P, Q \text{ polynomials}, Q \text{ convex in } y, \cY \text{ compact, convex, semi-algebraic s.t. } \valbilevel = h\}
\end{equation}
where $\coco$ stands for ``compact convex''. Using \Cref{theorem:compact=lsc+bc}, we get an immediate relation:
\begin{equation*}
    \begin{aligned}
        \cC_{\bounded} \subseteq \cC_{\coco} \subseteq \SA \cap \LSC \cap \BC.
    \end{aligned}
\end{equation*}
since bounded boxes are compact, convex, and semi-algebraic. Again, it is natural to investigate whether these inclusions are strict. Our following result shows that the second inclusion (in the above equations) is actually an equality.

\begin{theorem}[Value functions of polynomial bilevel programming with convex, compact and semi-algebraic lower-level constraints]
    \label{theorem:compact-convex=LSC}
    Any semi-algebraic, lower semicontinuous function that is bounded on compact sets is the value function of a polynomial bilevel problem with lower objective convex over a compact, convex, and semi-algebraic set. In particular,
    \begin{equation*}
        \begin{aligned}
            \cC_{\textup{\coco}} &= \SA \cap \LSC \cap \BC.
        \end{aligned}
    \end{equation*}
\end{theorem}
\begin{proof}
    This proof is based on the following fact: one can equivalently reformulate the lower-level polynomial optimization problem in the proof of \Cref{theorem:compact=lsc+bc}, by a convex optimization problem with a compact, convex, and semi-algebraic feasible set. {This reformulation leads to a modified argmin correspondence in the lower level which does not change the value function overall}. %Thus, we can replace the lower-level polynomial optimization problem in the proof of \Cref{theorem:compact=lsc+bc}.

    Let us provide details: consider the following optimization problem:
    \begin{equation*}
        \underset{\omega \in \Omega}{\text{Minimize}} \quad F(\omega),
    \end{equation*}
    where $\Omega \subseteq \RR^n$ is a compact, semi-algebraic set and $F$ is a polynomial. Let $d = \deg(F)$ be the highest degrees of a monomial (a product of powers of variables with nonnegative integer exponents $x_1^{d_1}\ldots x_n^{d_n}, d_i \in \NN, \forall 1 \leq i \leq n$) of $F$. Consider the function $M_d: \Omega \to \RR^{K_d}$ that maps a point $\omega \in \RR^n$ to the vector of monomials up to degree $d$ (the constant $K_d = {n + d \choose n}$ is the number of such monomials). Since $F$ is a polynomial, it can be written as: $F(\omega) = c^\top M_d(\omega)$ for some vector $c \in \RR^{R_d}$, i.e., $F(x)$ is linear w.r.t $M_d(\omega)$. Thus, the original polynomial optimization problem can be written equivalently as:
    \begin{equation*}
        {\text{Minimize}} \quad c^\top \lambda \quad \text{such that} \quad \lambda \in \conv{M_d(\Omega)},
    \end{equation*}
    where $M_d(\Omega) \subseteq \RR^{R_d}$ is the image of $\Omega$ via the map $M_d$ and $\conv{\cdot}$ is the convex hull of a set. Since $\Omega$ is compact, so are $M_d(\Omega)$ and its convex hull. The semi-algebraicity can be argued similarly (using Carathéodory's theorem for convex hull \cite[Theorem $0.0.1$]{vershyninHighdimensionalProbabilityIntroduction2018} and \Cref{theorem:projection-theorem}). Thus, this new formulation has a linear objective function with a compact, convex, and semi-algebraic feasible set.  

    Now, we will plug this reformulation into the construction in the proof of \Cref{theorem:compact=lsc+bc}. Consider $f \in \SA \cap \LSC \cap \BC$. Using the construction in \eqref{eq:construction-compact-lsc}, we remind readers that there exists two functions $P(x,t,z)$ and $Q(x,t,z)$ ($x$ is the upper-level variable, $t,z$ are lower-level variables) and a bounded box $\cY$ such that:
    \begin{enumerate}
        \item The value function of the associated bilevel optimization problem equals $f$.
        \item The function $P(x,t,z) = t$.
        \item For all $x$, the minimum value of $y$ such that there exists $z$ satisfying $(t,z) \in \argmin Q(x, \cdot, \cdot)$ is $f(x)$.
    \end{enumerate}
    Using this information, we can construct a new bilevel optimization problem as follows: instead of using $t,z$ as lower-level variables, we will use $\lambda$, a (vector-valued) variable representing all the monomials of the concatenation $(t,z)$ up to the degree of interest. We write $\lambda[t]$, for the coordinate entry of $\lambda$ corresponding to the degree-one monomial $t$. The upper-level, lower-level functions, and the lower-level feasible set of the new bilevel problem are respectively given by:
    \begin{equation*}
        \begin{aligned}
            P'(x,\lambda) &= \lambda[t], \\
            Q'(x,\lambda) &= c_Q^\top \lambda\\
            \cY' &= \conv{M_d(\cY)},
        \end{aligned}
    \end{equation*}
    where $c_Q$ is the vector containing the coefficients of the monomials of $Q$. For a given $x$, consider $\lambda^\star$ is an element of $\conv{M_d(\cY)}$ that satisfies:
    \begin{enumerate}
        \item $\lambda^\star \in \argmin_{\lambda \in \cY'} Q'(x, \lambda)$.
        \item $\lambda[t]^\star$ attains the minimum value among elements in $\argmin_{\lambda \in \cY'} Q'(x, \lambda)$ (since $P'(x,\lambda) = \lambda[t]$ and we are considering optimistic bilevel problems).
    \end{enumerate}
    Again, by Carathéodory's theorem \cite[Theorem $0.0.1$]{vershyninHighdimensionalProbabilityIntroduction2018}, an element of $\cY'$ must be written as a convex combination of at most $C = \dim(\lambda) + 1$ ($\dim(\lambda)$ is the dimension of $\lambda$) elements of $M_d(\cY)$. Therefore, there exists $\lambda_1, \ldots, \lambda_C \in M_d(\cY)$ and $c_1, \ldots, c_C \geq 0, \sum_{i = 1}^C c_i = 1$ such that:
    \begin{equation*}
        \lambda^\star = \sum_{i = 1}^C c_i\lambda_i.
    \end{equation*}
    By the linearity of $Q'$ and $P'$ with respect to $\lambda$, we can conclude that for all $i = 1, \ldots, C$, $\lambda_i \in \argmin_{\lambda \in \cY'} Q'(x, \lambda)$ and $\lambda_i[t]$ also attains minimum value among the elements of $\argmin Q'(x, \cdot)$. Since $\lambda_i \in {M_d(\cY)}$, $\lambda_i = M_d((t_i, z_i))$ where $(t_i, z_i) \in \argmin_{(t,z) \in \cY} Q(x, t, z)$ and $t_i$ is the smallest possible such value. It implies that $t_i = f(x), \forall i = 1, \ldots, C$. Hence, 
    $$P'(x, \lambda^\star) = \lambda^\star[t] = \sum_{i = 1}^C c_i\lambda_i[t] = \left(\sum_{i = 1}^C c_i\right)f(x) = f(x).$$
    This concludes the proof.
\end{proof}

Let us emphasize that the constraint set resulting from the proof of this result does not have an explicit construction and only allows for a looser control of the dimensionality, unlike previous results. Indeed, although $\cY$ is convex, compact and semi-algebraic, we do not know how to \emph{explicitly and efficiently} represent it (using polynomial equalities and inequalities). Such representation is important in polynomial optimization \cite{bach2023exponential} and closely related to (but not quite the same as) the SOS relaxation \cite{lasserre2001global,Parrilo2003SemidefinitePR}. We refer readers to \cite{bach2023exponential} for a more dedicated discussion. Note that the same idea (if one allows $\cY$ to be an arbitrary convex, closed semi-algebraic set) does not work for unbounded cases because the convex hull of a closed set is not necessarily closed. We did not find a way around this issue and leave this question for future work. We also leave open the possibility of obtaining similar results for simple lower-level constraints set $\cY$ such as balls or boxes. 

{Finally, while \Cref{theorem:convex-compact-bilevel-class} and \Cref{theorem:compact-convex=LSC} appear very similar, there is a subtle difference in the settings, and consequently, in the final results: $\cY$ in \Cref{theorem:convex-compact-bilevel-class} is a bounded box whereas $\cY$ in \Cref{theorem:compact-convex=LSC} is an arbitrary compact, convex and semi-algebraic set. The additional freedom to choose $\cY$ leads to two different results: \Cref{theorem:convex-compact-bilevel-class} shows an inclusion while \Cref{theorem:compact-convex=LSC} shows an equality. Note that \Cref{lem:constrained-convex-construction} cannot be used to treat the bounded case as it involves an unbounded box. Note also that the technique used in \Cref{theorem:convex-compact-bilevel-class} does not generalize to the unbounded case in \Cref{theorem:convex-closed-bilevel-class} because the convex hull of an unbounded set may not be closed, and adding a closure operation would break the Carathéodory representation which we employ in the proof of \Cref{theorem:compact-convex=LSC}. We leave a more precise characterization in \Cref{theorem:convex-closed-bilevel-class} and \Cref{theorem:convex-compact-bilevel-class} open for future research.}

%Another idea is to take the closure of the convex hull as in \cite{bach2023exponential}. This, however, does not work either since we cannot ensure that the coefficient of $\lambda$ corresponding to the variable $y$ equals $f(x)$ (cf. \Cref{tab:variables-semi-algebraic-function}). 

\subsection{Summary of the results and pessimistic bilevel problems}
\label{subsection:summary_results}
\Cref{tab:expressivity} summarizes below the various representation results mentioned above.

\begin{table}[htbp]
    \scriptsize
	\begin{tabular}{>{\centering\arraybackslash}m{3.2cm}>{\centering\arraybackslash}m{3.3cm}|>{\centering\arraybackslash}m{4.0cm}>{\centering\arraybackslash}m{2.7cm}}
            \toprule
            \multicolumn{2}{c}{\textbf{Type of bilevel problem}} & 
            \multicolumn{2}{c}{\textbf{Expressivity result}}\\
            \toprule
		\textit{Lower-level objective}& \textit{Lower-level constraints} & \textit{Expressible functions}& \textit{Value function class}\\
		\toprule	
        Nonconvex & unbounded box & (\textbf{SA}) & $=$ \\
        \midrule
        Nonconvex & bounded box & (\textbf{SA}) + (\textbf{LSC}) + (\textbf{CB}) & $=$ \\
	\midrule
        Convex & unbounded box & (\textbf{PP}) & $\subset$ \\
        \midrule
        Convex  & bounded box & (\textbf{PP}) + (\textbf{LSC}) & $\subset$ \\
        \midrule
        Convex  & convex, compact, semi-algebraic & (\textbf{SA}) + (\textbf{LSC}) + (\textbf{CB}) & $=$ \\
		\bottomrule
	\end{tabular}
    \caption{Summary of all expressivity results. Abbreviations (\textbf{PP}): \emph{piecewise polynomial}, (\textbf{SA}): \emph{semi-algebraic}, (\textbf{LSC}): \emph{lower-semicontinuous}, (\textbf{CB}): \emph{bounded on every compact}. All our results state that the class of expressible functions (third column) is contained in the class of value functions corresponding to the considered type of bilevel problem. The last column indicates whether this inclusion is actually an equality or not.\label{tab:expressivity}}
\end{table}

As explained in our introduction, we focused on the optimistic bilevel formulation. Let us emphasize that all our results have counterparts for the pessimistic bilevel formulation. Indeed, it can be checked that all our constructions can be adapted to treat this situation similarly. In particular, we have the following.
\begin{itemize}
    \item Both representation results involving unbounded boxes in \Cref{tab:expressivity} hold true also for the pessimistic versions \eqref{eq:pessimistic-bilevel} and actually for any specification of the selection process in the lower level argmin. Indeed, all these constructions involve unique minimizers in the lower level.
    \item For the results involving bounded constraint sets, our constructions can be adapted to treat the pessimistic formulation \eqref{eq:pessimistic-bilevel} similarly. The corresponding representation results would be exactly the same, with the notion of lower-semicontinuity being replaced by that of upper-semicontinuity, as the value function of pessimistic bilevel problems with compact lower level constraints are upper semicontinuous.
\end{itemize}
We do not include these results explicitly for simplicity, but in the context of pessimistic bilevel problems, all the results of \Cref{tab:expressivity} hold true with upper-semicontinuity replacing lower-semicontinuity.

\section{Computational hardness of polynomial bilevel optimization}

%\jer{Given our usual audience I would start by what we have called expressivity (don't like that word too much, it is a positive word to express negative results). ``Geometric hardness"?}

%\jer{Put one of these usual bubble figures showing P/NP/ etc... which allows one to catch in a glance how hard is bilevel}

\subsection{Preliminaries on the polynomial hierarchy}
\label{subsec:premilinary-poly-hierarchy}
This section reminds readers of the polynomial hierarchy, a classification of problems based on their ``hardness'' in computational complexity. We also discuss the subset sum interval problem -- a classical $\existleveltwo$-hard problem that will play an important role in our analysis. 
\subsubsection*{Polynomial hierarchy} 
A classical definition of $\mathbf{P}$ and $\mathbf{NP}$, two important concepts of complexity theory, is based on the Turing machine models: $\mathbf{P}$ and $\mathbf{NP}$ is the set of \emph{decision} problems that can be solved in polynomial time using deterministic and non-deterministic Turing machines, respectively. Alternatively, one can define a problem belonging to $\mathbf{NP}$ if for any instance whose answer is \emph{yes}, there exists a proof verifiable in polynomial time (using a deterministic Turing machine). In other words, verifying a \emph{yes} instance of a $\mathbf{NP}$ problem is a $\mathbf{P}$ problem. The negated class of $\mathbf{NP}$ - $\mathbf{coNP}$ - contains those whose instances with \emph{no} answer can be verifiable in polynomial time. Thus, different from the definitions based on the non-deterministic Turing machine, one can define two classes $\mathbf{NP}$ and $\mathbf{coNP}$ based on $\mathbf{P}$.

A natural generalization of this approach gives us the polynomial hierarchy. Following \cite[Theorem 3.1]{STOCKMEYER19761}, one can define the complexity class $\Sigma_k^p$ as the set of decision problems that can be written in the form:
\begin{equation*}
    (\exists y_1), (\forall y_2), (\exists y_3), \ldots, (\cQ_k y_k), R(x, y_1, \ldots, y_k),
\end{equation*}
where the quantifiers $\cQ_\ell \in \{\exists, \forall\}, 1 \leq \ell \leq k$ alternate and $R$ is a boolean formula that can be evaluated in polynomial time (or equivalently, it is a problem in $\mathbf{P}$). In particular, $\Sigma_1^p = \mathbf{NP}$.

Analogously, one can exchange the role of $\exists$ and $\forall$ in the definition of the class $\Sigma_k^p$ to define the class $\Pi_k^p$ as:
\begin{equation*}
    (\forall y_1), (\exists y_2), (\forall y_3), \ldots, (\cQ_k y_k), R(x, y_1, \ldots, y_k).
\end{equation*}
Similarly, we also have $\Pi_1^p = \mathbf{coNP}$. By convention, $\mathbf{P} = \Sigma_0^p = \Pi_0^p$. Thus, similar to the relation $\mathbf{P} \subseteq \mathbf{NP} \cap \mathbf{coNP}$, we have the generalized version:
\begin{equation*}
    \begin{aligned}
        \Sigma_k^p \subseteq \Sigma_{k+1}^p\cap \Pi_{k+1}^p \qquad \text{ and } \qquad \Pi_k^p \subseteq \Sigma_{k+1}^p \cap \Pi_{k+1}^p.\\
    \end{aligned}
\end{equation*}
For any $ k\in \NN$, it remains unknown whether $\Sigma_k^p \neq \Sigma_{k+1}^p$ or $\Sigma_k^p = \Sigma_{k+1}^p$, and similarly for $\Pi_{k}^p$ and $\Pi_{k+1}^p$. Furthermore, if there is equality for a given $k_0$, then there is equality for all $k \geq k_0$ \cite[Theorem 5.4]{arora2009computational}. This is called the \emph{collapse} of polynomial hierarchy at the $k_0$-th level, a possibility which is considered unlikely and is often used as an assumption in complexity theoretic proofs, see discussions in \cite[Chapter 5]{arora2009computational}).
Readers can view an illustration of the polynomial hierarchy and the relations between their components in \Cref{fig:poly-hierarchy}. 

\begin{figure}[bhtp]
    \centering
    \includegraphics[scale = 0.4]{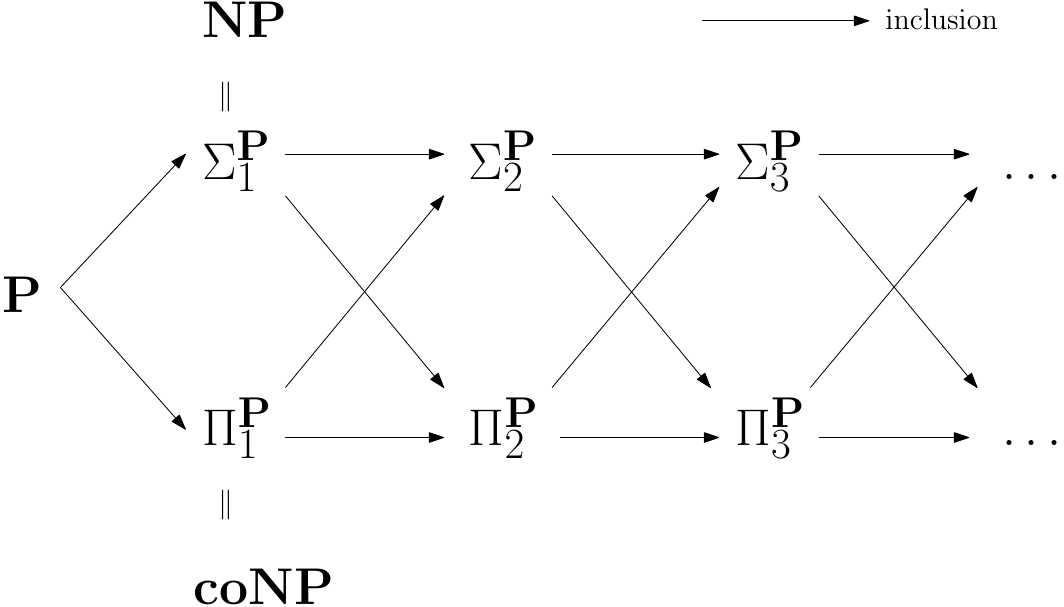}
    \caption{Diagram describes the polynomial hierarchy. Arrows represent the inclusion relation.}
    \label{fig:poly-hierarchy}
\end{figure}

A famous tool to study the relation between $\mathbf{P}$ and $\mathbf{NP}$ is the notion of $\mathbf{NP}$-hardness. A problem $\cA$ is called $\mathbf{NP}$-hard if for any problem $\cB$ in class $\mathbf{NP}$, there exists a transformation (also known as \emph{reduction}) running in polynomial time that turns an instance of $\cB$ into that of $\cA$ and both share the same answer (\emph{yes} and \emph{no}). As such, if one comes up with a polynomial algorithm for $\cA$, then he can solve all $\mathbf{NP}$ problems in polynomial time. On the other hand, to prove that a problem $\cA$ is $\mathbf{NP}$-hard, it is sufficient to construct a polynomial reduction from any $\mathbf{NP}$-hard problem to $\cA$.

Analogously, we can define $\Sigma_k^p/\Pi_k^p$-hardness using the same principle for any $k \in \N$. We will use the concept of $\Sigma_2^p$ hardness. If a problem is $\Sigma_2^p$-hard and the polynomial hierarchy does not collapse at the first level, i.e, $\mathbf{NP} \neq \Sigma_{2}^p$, then the intrinsic computational complexity of this problem is strictly higher than any problem in the class $\mathbf{NP} \cup \mathbf{coNP}$. In the following, we will show that polynomial bilevel optimization is $\existleveltwo$-hard. This highlights the possibility that bilevel polynomial optimization is actually much harder than problems in $\mathbf{NP}$ which are themselves computationally hard. Note that it is a common assumption in computational complexity theory that the polynomial hierarchy does not collapse at all, and specially not at the first level (see discussions in \cite[Section 5.2.1]{arora2009computational}).

\subsubsection*{The subset sum interval problem}
To show the $\existleveltwo$-hardness of polynomial bilevel optimization, we need to perform a reduction to a known $\existleveltwo$-hard problem. We choose the \emph{subset sum interval problem} \cite{caprara2014study}, defined as follows:
\begin{problem}[Subset sum interval problem]
    \label{prob:subset-sum-interval}
    Given a finite number of positive integers $q_1, \ldots, q_k$, and two positive integers $R$ and $r$ with $r \leq k$, decide whether there exists an integer $R \leq S \leq R + 2^r - 1$ such that none of the subsets $I \subseteq \{1, \ldots, k\}$ satisfies $\sum_{i \in I}q_i = S$.  
\end{problem}

If $r > k$, the problem is trivial because there are at most $2^k$ possible values of $\sum_{i \in I} q_i$ but there are already $2^r$ possible values for $S$. Using logic notations, \Cref{prob:subset-sum-interval} is equivalent to deciding the correctness of this first-order logic formula:
\begin{equation*}
    \exists S, \forall I \subseteq \{1, \ldots, k\}, R \leq S \leq R + 2^r - 1 \text{ and } S \neq \sum_{i \in I} q_i.
\end{equation*}
The $\existleveltwo$-hardness of this problem follows from \cite{eggermont2012motion}, see \Cref{sec:subsetSumSigma2pHard} for more details.

\subsection{Polynomial bilevel optimization is $\existleveltwo$-hard}

The main result of this section is to prove the $\existleveltwo$-hardness of the decision version of the polynomial bilevel problem.
\begin{theorem}[Hardness of polynomial bilevel optimization]
    \label{theorem:hardness2}
    Given two polynomials $P$ and $Q$ of degree at most five, a bounded box lower-level feasible set $\cY$ and a constant $c$, deciding whether the optimal value of the corresponding bilevel problem of $(P,Q,\cY)$ is \emph{strictly} smaller than $c$ is $\existleveltwo$-hard.
\end{theorem}
\Cref{theorem:hardness2} is an immediate result of the following lemma.
\begin{lemma}
    \label{lemma:hardness-bounded}
    Given an instance of \Cref{prob:subset-sum-interval}, there exist two polynomials of degree five $P$ and $Q$ whose coefficients are integers in the interval $[-2(M^2 + 1), 2(M^2 + 1)]$ where $M = \max(\max_{i = 1}^k q_i, 2^{r-1}, R)$ and variable $(\firstvar,\secondvar)$, where $\firstvar \in \RR^{r}, \secondvar \in \RR^{k + 1}, \cY = [0,1]^{k+1}$, such that the optimal value of \eqref{eq:optimistic-bilevel}  
    \iffalse
    the following optimistic polynomial bilevel problem:
    \begin{equation*}
        \begin{aligned}
            &\min_{\firstvar} && \valoptim(\firstvar):= \min_{(x, y(x))} P(\firstvar, \secondvar(\firstvar))\\
            &\text{s.t.} && \secondvar(\firstvar) \in \argmin_{\secondvar \in \cY} Q_\firstvar(\secondvar) := Q(\firstvar, \secondvar) 
        \end{aligned}
    \end{equation*}
    \fi
    is strictly smaller than zero if the answer for the instance of \Cref{prob:subset-sum-interval} is YES, and exactly equal to zero if the answer is NO. 
\end{lemma}

Before proving the lemma, we argue that the ``description'' of the polynomials $P,Q$ in \Cref{lemma:hardness-bounded} is at most polynomial w.r.t. the size of the inputs of \Cref{prob:subset-sum-interval}. Indeed, the polynomials $P,Q$ have degree five and all coefficients can be expressed by $\log M$ bits. We remark that the representation of each coefficient might use $O(r)$ bits, from the bound $2^{r-1}$, which is exponential w.r.t. the number of bits representing the number $r$ itself. In addition, we have polynomials of $r$ variables, which represent a number of coefficient polynomial in $r$. This dependency in $r$ is still polynomial overall because any bit representation of \Cref{prob:subset-sum-interval} must use at least $k \geq r$ bits to represent to numbers $q_i, i = 1, \ldots, k$. Regarding the other quantities appearing in the definition of $M$, it is clear that an integer smaller than $\max_{i = 1}^k q_i$ has a bit representation of size bounded by that of the collection $q_1,\ldots, q_k$ and similarly for $R$.

\begin{proof}[Proof of {\Cref{lemma:hardness-bounded}}]
        We will use $x \in \RR^{r}$ and $y = (z, t), z \in \RR^r, t \in \RR^{k}$ (which implies $y \in \RR^{r + k}$) to indicate the upper-level and lower-level variables respectively. We also use $x_i, z_i, t_i \in \RR$ to indicate the $i$th coordinate of the variables $x, z, t$ respectively. For $x \in \RR^r$, define:
    \begin{equation*}
        \label{eq:binary_representation}
        F(x) := R + \sum_{j = 1}^r 2^{j-1}x_j,
    \end{equation*}
    a linear combination (thus, a polynomial) of $x$. Intuitively, ${x}$ is the binary {encoding} of a number in the range $[R, R + 2^r - 1]$ if all variables $x_j$ are binary. 
    
    Moreover, for $t \in \RR^k$, we also define:
    \begin{equation*}
        \label{eq:subset_representation}
        G(t) := \sum_{i = 1}^k t_iq_i,
    \end{equation*}
    where $q_i$ is the integers that appear in the given instance of \Cref{prob:subset-sum-interval}. Intuitively, $G(t)$ is equal to the sum of some subsets of $\{q_i, i = 1, \ldots, k\}$ (if all variables $t_i$ are binary). Finally, we define:
    \begin{equation*}
        H(x,t) = \left(\sum_{j = 1}^k (t_i(1 - t_i))^2\right) + (F(x) - G(t))^2
    \end{equation*}
    Note that $H(x,t)=0$ if and only if $t$ is a binary vector and $F(x) = G(t)$. This mimics the situation where a number $F(x) \in [R, R + 2^r - 1]$ can be written as a subset-sum of the array $q$.
    
    The polynomials $P, Q$ and the lower-level feasible set can be constructed as follows: We constrain the lower-level variable $\secondvar = (z,t) \in \cY = [0,1]^{k+1}$, or equivalently, $z \in [0,1], t \in [0,1]^k$. 
    We define $P$ and $Q$ as follows:
    \begin{equation}
        \label{eq:construction-P-Q}
        \begin{aligned}
            P &= 1 - z(H(\firstvar,t) + 1),\\
            Q &= \underbrace{z\left[(z-1)^2 + \sum_{i = 1}^r \left((1 - x_i)x_i\right)^2\right]}_{L(x,z)} + H(\firstvar,t).\\
        \end{aligned}
    \end{equation}
    By construction, $\deg(P) = \deg(Q) = 5$ as stated. In addition, it can be shown that the coefficients of monomials of $P$ and $Q$ have their absolute values bounded by $2(M^2+1)$. Indeed, all the coefficients of $L(x,z)$ in \eqref{eq:construction-P-Q} belong to $\{0, 1, -2\}$. A direct calculation also shows that absolute values of the coefficients of $H(x,t)$ are bounded by $2(M^2 + 1)$. More importantly, these two polynomials do not have any common monomial ($H$ does not have variable $z$ but every monomial of $L$ has at least one $z$). Thus, we can conclude that the coefficients of $P$ and $Q$ lie in the interval $[-2(M^2 + 1), 2(M^2 + 1)]$.

    Due to the separation of $z$ and $t$, given a fixed $x$, we have:
    \begin{equation}
        \label{eq:structure_of_optimum}
        \argmin Q(x, \cdot, \cdot) = \argmin L(x, \cdot) \times \argmin H(x, \cdot).
    \end{equation}
    where $\times$ is the Cartesian product between two sets. The key idea in this construction is based on a sequence of observations.

    \paragraph{Observation 1} With the construction of $P$ and $Q$ as in \eqref{eq:construction-P-Q}, we have:
    \begin{enumerate}
        \item If $\firstvar \notin \{0,1\}^r, \valbilevel(\firstvar) = 1$.
        \item If $\firstvar \in \{0,1\}^r, \valbilevel(\firstvar) \leq 0$. 
    \end{enumerate}
    
    Indeed, we have: $L(x,z) \geq 0$ because $z \in [0,1]$ and the other factor is a sum of squares. Thus, $z = 0$ is always a minimizer of $L(x,\cdot)$ on $[0,1]$. Another possibility to attain the global minimum value of $L(\firstvar, z)$ is to have $x \in \{0,1\}^r$ and $z = 1$. Therefore,
    \begin{equation}
        \label{eq:optimal-z-1}
        z(x):= \arg\min_z L(x,z) = \begin{cases}
            0 & \text{ if } x \notin \{0,1\}^r\\
            \{0,1\} & \text{ otherwise } \\
        \end{cases} \;\implies\; \valbilevel(\firstvar) = \begin{cases}
            1 & \text{ if } x \notin \{0,1\}^r\\
            -H(\firstvar, t(x)) & \text{ otherwise}
        \end{cases}.
    \end{equation}
    where $(z(\firstvar), t(\firstvar)) \in \argmin Q(\firstvar, \cdot, \cdot)$. Since $H$ is a sum of squares, we proved the first observation.

    To finish the proof, we need a second observation.
    \paragraph{Observation 2} With the construction of $P$ and $Q$ as in \eqref{eq:construction-P-Q}, we have:
    \begin{enumerate}
        \item \textbf{First case}: If there exists an integer $S \in [R, R + 2^r - 1]$ such that it is not equal to the sum of any subset of $\{q_i, i = 1, \ldots, k\}$, then  $\inf_x \valbilevel(x) < 0$. 
        \item \textbf{Second case}: Otherwise, $\valbilevel(x) = 0, \forall x \in \{0,1\}^{r}$. Hence, $\min_x \valbilevel(x) = 0$.
    \end{enumerate}
    Due to the previous observation, it is sufficient to consider only binary inputs $x$. In that case, $\valbilevel(x) \in -H(x,t(x)) = -\min_{t \in [0,1]^k} H(x,t)$ (cf. \Cref{eq:optimal-z-1}).
    In addition, $F(x)$ is an integer belonging to the interval $[R, R + 2^r - 1]$. We consider two cases one by one:
    \begin{enumerate}[leftmargin=*]
        \item \textbf{First case}: We choose $x \in \{0,1\}^r$ such that $S = F(x)$. Moreover, due to the property of $S$, $H(x,t) > 0$ because $H(x,t) = 0$ if and only if $t$ is a binary vector and $F(x) = G(t)$. However, that will be equivalent to the statement that $S$ equals the sum of some subsets of $\{q_i, i = 1, \ldots, k\}$, a contradiction. Therefore, $\valbilevel(x) = -\min_{t \in [0,1]^k} H(x,t) < 0$. 
        \item \textbf{Second case}: In this case, for any $R \leq S = F(x) \leq R + 2^r - 1$, there exists a binary vector $t$ such that $F(x) = G(t)$. Thus, in the lower-level problem, when we minimize $H(x,t)$, we will get a binary vector $t(x)$ such that $H(x,t(x)) = 0$. It allows us to conclude that $\valbilevel(x) = -H(x,t(x)) = 0, \forall x\in\{0,1\}^r$.   
    \end{enumerate}
    Combining two cases yields the proof.
\end{proof}
For interested readers, the whole proof of \Cref{lemma:hardness-bounded} is to find two polynomials whose corresponding optimistic bilevel problem equals to the lower-semicontinuous function $\valbilevel$ in \eqref{eq:optimal-z-1}. Thus, our construction is similar to the proof of \Cref{theorem:compact=lsc+bc}, with some simplification adapted to the structure of $\valbilevel$ to minimize the degrees of $P,Q$. This proof cannot be extended for the pessimistic version since a pessimistic bilevel problem with a bounded box lower-level constraint has an upper-semicontinuous value function. Finally, one might wonder if we can construct a ``difficult'' instance with a convex lower-level problem. Nevertheless, our technique in \Cref{theorem:compact-convex=LSC} cannot be applied since the reduction is not guaranteed to be polynomial.

\bigskip

\section{Conclusion} 
{
The pathological examples we provided show that neither rigidity nor smoothness makes bi-level programming tractable, whether it is from a geometrical or computational complexity. This calls for the identification of new classes of problems that possess a strong overall regularity and which are amenable to optimization. In particular, it encourages the search for new, easily verifiable qualification conditions and a thorough study of their interaction with usual regularity assumptions, and with the behavior of solution algorithms.}

\begin{appendices}
\section{The difficulty of smooth bilevel optimization}
\label{appendix:regularity-not-good-enough}
In this section, we provide the proof of \Cref{prop:reduction-lsc}. The proof uses the Whitney representation of closed set \cite[Section 3.10, 2]{lafontaine2015introduction} \cite[Theorem 2.29]{lee2000smoothmanifolds}.

%\ed{Theorem 2.29 in Lee's introduction to smooth manifolds (second edition). It is not attributed to Whitney though.}

%\TLc{Where to cite the result of Whitney?}
\begin{theorem}[Whitney representation of closed sets]
    \label{theorem:whitney-representation}
    Any closed set of $\RR^n$ is the set zeros of a smooth function $f: \RR^n \to \RR$. 
\end{theorem}

We remind readers of an important property of lower (resp. upper) semicontinuous functions.
\begin{proposition}
    \label{prop:property-lsc}
    A function $f:\RR^n \to \RR$ is lower semicontinuous if and only if its epigraph is closed where the epigraph of $f$ is defined as:
    \begin{equation*}
        \begin{aligned}
            \mathrm{epi}(f) &:= \{(x,\alpha) \mid f(x) \leq \alpha\} \subseteq \RR^{n+1}.
        \end{aligned}
    \end{equation*}
\end{proposition}

\begin{proof}[Proof of \Cref{prop:reduction-lsc}]
    If $f$ is lower semicontinuous, the epigraph of $f$ -- $\mathrm{epi}(f)$ -- is closed. Using \Cref{theorem:whitney-representation}, there exists a smooth function $h: \RR^{n+1} \to \RR$ such that its zeros set equals $\mathrm{epi}(f)$. 
    
    Since the domain $\mathrm{dom}(f) := \{\firstvar \in \RR^n \mid f(\firstvar) < +\infty\}$ is closed (due to our assumption), by \Cref{theorem:whitney-representation}, there exists $g: \RR^n \to \RR$ such that its zeros equals $\mathrm{dom}(f)$.
    
    Consider the upper-level and lower-level variables $\firstvar \in \RR^n$ and $\secondvar = (\secondvar_1, \secondvar_2, \secondvar_3) \in \RR^3$, define:
    \begin{equation*}
        \begin{aligned}
            P &:= \secondvar_1\\
            Q &:= h(\firstvar,\secondvar_1)^2 + \underbrace{\left(g(\firstvar)\secondvar_2\right)^2 + (1 - \secondvar_2\secondvar_3)^2}_{p(\firstvar, \secondvar_2, \secondvar_3)}.
        \end{aligned}
    \end{equation*}
    We argue that with this choice of $P,Q$, the function $\valbilevel$ of the optimistic version equals $f$. Indeed, due to the separation of $\secondvar_1$ and $(\secondvar_2, \secondvar_3)$, we have:
    \begin{equation*}
        \argmin_\secondvar Q = \argmin_{\secondvar_1} h(\firstvar, \cdot) \times \argmin_{\secondvar_2, \secondvar_3} p(\firstvar, \cdot, \cdot).
    \end{equation*}
    Consider two cases:
    \begin{enumerate}
        \item If $\firstvar \in \mathrm{dom}(f)$: $\Theta(\firstvar) = \argmin Q(\firstvar, \cdot)$ is the set of $\secondvar$ such that $h(\firstvar,\secondvar_1) = 0$ and $\secondvar_2\secondvar_3 = 1$. Therefore, $\Theta(x) = \mathrm{epi}(f) \cap \{x\} \times \RR$. Since we deal with the optimistic version, $\valbilevel(\firstvar) = \min \{\secondvar \mid \secondvar \in \mathrm{epi}(f) \cap \{x\} \times \RR\} = f(x)$ by definition of epigraph. This concludes the proof. 
        \item If $\firstvar \notin \mathrm{dom}(f)$, then $g(\firstvar) \neq 0$. We claim that the set $\argmin_{\secondvar_2, \secondvar_3} p(\firstvar, \cdot, \cdot)$ is empty and that will conclude the proof. Indeed, the infimum of $p$ is zero by taking $\secondvar_2 \to 0$ and $\secondvar_3 = 1 / \secondvar_2$. However, this infimum cannot be attained since both squares cannot equal zero simultaneously. That concludes the proof.
    \end{enumerate}
\end{proof}
\noindent \textbf{Discussion} \Cref{prop:reduction-lsc} highlights the impossibility of dealing with the general $ C^\infty$ bilevel problem and raises concerns about the very meaning of a solution algorithm for such problems. Indeed, to have a sense of Proposition \ref{prop:reduction-lsc}, one can observe that the following monstrous univariate functions admit a representation as a bilevel program with smooth data: 
\begin{itemize}
    \item The \emph{negative} of Thomae's popcorn function \cite[Example 5.1.6h)]{bartle2011realanalysis} is lower semicontinuous and discontinuous on $\mathbb{Q}$, with value $0$ on irrationals and global minimum at $1/2$.
    \item The characteristic function of the Smith-Voltera-Cantor set \cite[Definition 11.1.10]{bartle2011realanalysis} ($1$ on the set, $0$ outside). This is a closed set, its characteristic function is upper semicontinuous, but discontinuities have positive Lebesgue measure.
    \item The Weierstrass function \cite[Remark of Theorem 6.1.2]{bartle2011realanalysis} which is continuous but nowhere differentiable and does not have bounded variations.
    \item The Cantor staircase \cite[Section 6.5.3]{thomson2008elementary}, which is monotone, nonconstant, and almost everywhere differentiable with null derivative.
    \item Pathological Lipschitz function for which local minimizers form a dense subset \cite{loewen2000typical}, or for which subgradient sequences may fail to have a minimizing behavior \cite{daniilidis2020pathological,rios2022examples}.
\end{itemize}  
In particular, \emph{any method resembling a gradient algorithm on the value function may encounter insurmountable difficulties \textbf{even} if the problem data is arbitrarily smooth}.
% \begin{theorem}[Value functions of  polynomial bilevel programming]
%     \label{theorem:closed=SA-extended-valued}
%     Any extended-valued semi-algebraic function is the value function of a polynomial bilevel problem whose lower-level is unconstrained. In particular,
%     $$\cP^{\optim}_{\unbounded} = \cP^{\pessim}_{\unbounded} = \SA.$$
% \end{theorem}

\section{Connection to existing work}
\label{appendix:existing-works}
Our study provides an explanation for limited theoretical guarantees of many proposed algorithms of bilevel optimization, especially when the lower-level problem is not (strongly) convex. In the literature, many works proposed or analyzed algorithms based on automatic differentiation \cite{liu2021towards,liu2021valuefunction,arbel:hal-03869097,liu2020generic}. In this approach, one replaces the condition $\secondvar \in \Theta(\firstvar)$, cf. \eqref{eq:simple-bilevel-poly-optim}, by an algorithm $\cA$ minimizing $Q(\firstvar, \cdot)$. Intuitively, if the algorithm $\cA$ is differentiable w.r.t. to the upper variable $\firstvar$, then one can also calculate the gradient of the bilevel problem w.r.t. to the upper variable $\firstvar$ by the chain rule and use classical first-order methods. The difficulty of this approach is that for general nonconvex functions, most algorithms $\cA$ can only find stationary points or local minima of $Q(\firstvar, \cdot)$. Therefore, guarantees of these algorithms are established either under strong assumptions (e.g., regularity of $\Theta(\firstvar)$, uniform convergence of $\cA$) \cite{liu2021towards,liu2020generic} or for a relaxed model (e.g., $\secondvar$ is only required to be a stationary point) \cite{arbel:hal-03869097}. Another approach is based on smoothing techniques \cite{liu2021valuefunction,guihua2014solving} for which accumulation points of optimal solutions of the smooth approximation are the stationary points or minimizers of the original/approximate bilevel problems. Since the class of the value functions of general bilevel optimization problems can be large (e.g., \Cref{prop:reduction-lsc} and our following results), it is \emph{unsurprising} that standard convergence result of bilevel optimization might require much stronger assumptions to be established.

To avoid pathologies that will be shown in this work, another line of work of bilevel optimization focuses on establishing necessary conditions for locally optimal solutions. These works \cite{ye1995optimality,dempe2012sensitivity,Mordukhovich2020,hendrion2011calmness} are mostly based on the reformulation of bilevel problems into single-level ones using the so-called value function of the lower-level problem \cite{chen1995nonlinear}, which is constrained to be non-positive. In \cite{ye1995optimality}, the authors proposed to use the so-called \emph{partial calmness} to establish Karush-Kuhn-Tucker (KKT)-like necessary conditions for bilevel optimization. Their proposed qualification constraint plays an important role in this line of research since contrary to classical optimization, other popular constraint qualifications conditions such as linear independence (LICQ), Mangasarian-Fromovitz (MFCQ), and Slater's condition generally fail to hold %for the reformulated single-level problem 
(see \cite[Proposition 4.1]{ye1995optimality} for more discussion). The result in \cite{ye1995optimality} applies only if we view bilevel programming as a classical optimization problem that jointly minimizes $P$ w.r.t. to $(\firstvar,\secondvar)$ with a special constraint $\secondvar \in \Theta(\firstvar)$. To achieve necessary conditions of the local solutions of $\valbilevel(\firstvar)$, subsequent works required in addition inner semicontinuity of certain set-valued mappings related to $\Theta(\firstvar)$ \cite{dempe2012sensitivity,Mordukhovich2020,hendrion2011calmness}. While these constraint qualification conditions and assumptions are non-trivial and mathematically interesting, it is difficult to identify a class of bilevel problems that satisfy all of them. Our results might partly explain this difficulty, at least in the polynomial setting (see e.g., \Cref{theorem:closed=SA,theorem:compact=lsc+bc}). %\jer{Annoncer la motivation en premier lieu : les conditions de qualif permettent la bonne position et évitent les pathologies que nous constatons}

Our study on worst-case analysis is different but complementary to works devoted to algorithmic aspects of polynomial bilevel optimization. Using dedicated tools from polynomial optimization, in \cite{jeyakumar2016convergent}, the authors proposed semidefinite programming (SDP) relaxations for polynomial bilevel optimization with lower-level convex problems. More specifically, thanks to the lower-level convexity (and other reasonable constraint qualifications), one can reformulate polynomial bilevel problems into single-level constrained optimization with polynomial objectives and constraints, using KKT conditions. It allows using SDP relaxation techniques for solving polynomial optimization \cite{lasserre2001global}. When the lower-level problem is not convex, \cite{jeyakumar2016convergent,jiawang2017bilevel} proposed to solve an $\epsilon$-approximate version of the polynomial bilevel problem. Thus, the algorithm consists of two nested loops: the outer one solves the $\epsilon$-approximation problem and the inner one solves a sequential SDP relaxation corresponding to a fixed $\epsilon$. Cluster points of the sequence $\{(\firstvar_\epsilon, \secondvar_\epsilon)\}_{\epsilon > 0}$ (optimal solutions for the $\epsilon$-approximate problem) are the optimal solutions of the original problem (under certain assumptions). %\jer{A nouveau quel est le lien ?? L'annoncer préalablement}

%Our characterization of the geometric complexity of bilevel optimization is related to general representation results for semi-algebraic sets, for example in game theory. Indeed, \Cref{theorem:closed=SA} states that any semi-algebraic function is the value function of a polynomial bilevel problem \jer{attention a ne pas ouvrir la porte à des attaques, que va t on répondre à la question : comment comparez vous votre résultat à celui de Vigeral?}. An analogous result states that any nonempty compact semi-algebraic set is the set of Nash equilibrium payoffs of a finite game with $N \geq 3$ players \cite{vigeral2023characterization}. Many related characterizations of the sets of Nash equilibria and equilibrium payoffs described with the (semi-)algebraic notions are studied in \cite{vigeral2016semialgebraic,levy2015projections,datta2003universality,balkenborg2014universality}.

This work also investigates computational complexity, an aspect that is well-studied for the cases of (integer) linear bilevel programming. Indeed, it is known that, while linear bilevel optimization (i.e., $P,Q$ are linear, $\cY$ is a polyhedron) is $\textbf{NP}$-hard \cite{jeroslow1985polynomial,bard1991property,ben-ayed1990computational}, integer linear bilevel optimization (i.e., when several variables are constrained to be integers) belongs to the so-called $\existleveltwo$-hard problems  \cite{lodi2013bilevel,caprara2014study,jeroslow1985polynomial}. Essentially, the class $\existleveltwo$ is a generalization of the notions $\textbf{P}$ and $\textbf{NP}$ to capture the complexity of integer linear bilevel programming (see \Cref{subsec:premilinary-poly-hierarchy} for more details).
Moreover, if the polynomial hierarchy does not collapse in the first level, i.e. $\textbf{NP} \subsetneq \existleveltwo$, these results imply that the complexity of the $\existleveltwo$-hard integer linear bilevel optimization problem is much higher than that of any $\textbf{NP}$ complete problem. We prove the analog $\existleveltwo$-hardness of \emph{polynomial bilevel optimization}. This result is natural since single-level polynomial optimization itself is $\textbf{NP}$ hard (similarly as linear programming with integer constraints), and bilevel problems typically incur a complexity jump in the polynomial hierarchy.

\section{An alternative definition of semi-algebraic sets}\label{appendix:alternative-def}

\begin{proposition}    
    \label{lemma:stratification-elementary-set}
    Any semi-algebraic set $S \in \RR^n$ can be written as a finite union of disjoint basic semi-algebraic sets.
\end{proposition}

\begin{proof}
    By definition of semi-algebraic sets (cf. \Cref{def:semi-algebraic-set}), $S$ has the following form:
    \begin{equation}
        \label{eq:decomposition-si}
        S = \bigcup_{i \in I}  S_i \quad \text{where} \quad S_i:= \{x \mid P_i(x) = 0 \text{ and } Q_{ij}(x) > 0, j \in J\}.
    \end{equation}
    (We assume the same $J$ for all $i \in I$ for simplicity, but one can increase the number of inequalities in the definition of each $S_i$ to make this assumption valid).
    % Since $s_i, i = 1, \ldots, N$ are disjoint, it is sufficient to find a finite collections of disjoint basic elementary sets $\cT_\ell, \ell = 1, \ldots, K$ such that $s_i = \cup_{i = 1}^K \cT_\ell$. The collection satisfying the statement of the lemma is the union of such collections for all $s_i, i = 1, \ldots, N$. 

    To this end, let $\cH:= \{P_i, i = 1 \in I\} \cup \{Q_{ij}, (i,j) \in I \times J\}$ the set of all polynomials appearing in \eqref{eq:decomposition-si}. It is noteworthy that $\cH$ is finite, i.e., $|\cH| < +\infty$, thus we can write $\cH = \{h_1, \ldots, h_{|\cH|}\}$ where $h_i$ are polynomials. Consider the following collections of $3^{|\cH|}$ basic semi-algebraic sets, indexed by $\cI \in \{<, >, =\}^{|\cH|}$, and defined as:
    \begin{equation*}
        \cT_\cI := \{x \mid h_i \quad \Delta_i \quad 0, i = 1, \ldots, \cH\},
    \end{equation*}
    where $\Delta_i$ can receive three possible values $\{>, <, =\}$, encoded by the index $\cI$. These sets are semi-algebraic and they are inherently disjoint. Moreover, for each $\cT_\cI$ is either disjoint and included in $S_i$ for $i \in I$. Finally, the union of these sets equals to $\RR^n$. Therefore, there must exist a subset of indices $I \subseteq \{<, >, =\}^{|\cH|}$ such that:
    \begin{equation*}
        S = \bigcup_{\cI \in I} \cT_\cI.
    \end{equation*}
    This concludes the proof.
\end{proof}

\section{$\existleveltwo$-hardness of the subset sum interval problem}
\label{sec:subsetSumSigma2pHard}
The $\existleveltwo$-hardness of the subset sum interval problem is mentioned in \cite{caprara2014study}, and was credited to \cite{eggermont2012motion}. However, in \cite{eggermont2012motion}, the authors did not prove that the subset sum interval problem is $\existleveltwo$-hard. Instead, they proved that a problem called \emph{captive queen} is $\Pi^p_2$-hard. The goal of this section is to clarify the link the work of \cite{eggermont2012motion} and the subset sum interval problem: an intermediate result in \cite{eggermont2012motion} imply its $\existleveltwo$-hardness. But this implication is non-trivial and we describe the arguments here for self-containedness of our presentation. We do \emph{not} claim any scientific contribution as these results are due to \cite{eggermont2012motion,caprara2014study}. We thank Carvalho M. \footnote{\url{https://margaridacarvalho.org/}} for insightful discussions on this question.

We start with the quantified 3-CNF-SAT problem: given a boolean formula $\phi(X,Y)$ over variables $X$ and $Y$, in conjunctive normal form with each clause with exactly three literals, decide if $\forall X$, $\exists Y$, $\phi(X,Y)$ is true. This is a $\Pi_2^p$-complete problem \cite[Example 5.9]{arora2009computational}. 

Given such a $\phi$, \cite[Lemma 3.3 and Lemma 3.4]{eggermont2012motion} describe the construction of a polynomial time reduction to an instance of the so-called \emph{captive queen problem} \cite{eggermont2012motion} of the following forms: there exists $k$ integers $q_1, \ldots, q_k$ and two positive integers $R, r \leq k$ such that $\forall X$, $\exists Y$, $\phi(X,Y)$ holds, if and only if

\begin{equation*}
    \forall S, \exists I \subseteq \{1, \ldots, k\}, R \leq S \leq R + 2^r - 1 \text{ and } S = \sum_{i \in I} q_i.
\end{equation*}
This is the negation of the truth value of the subset sum interval problem. Since the 3-CNF-SAT problem is $\Pi_2^p$ complete, we have that negation of the subset sum interval problem is $\Pi_2^p$-complete, and the problem itself $\Sigma_2^p$-complete \cite[Section 5.1]{arora2009computational}.

\section*{Acknowledgements}
JB, TL, EP thank AI Interdisciplinary Institute ANITI funding, through the French ``Investments for the Future -- PIA3'' program under the grant agreement ANR-19-PI3A0004, Air Force Office of Scientific Research, Air Force Material Command, USAF, under grant numbers FA8655-22-1-7012. JB, EP and SV acknowledge support from ANR MAD. JB and EP thank TSE-P and acknowledge support from ANR Chess, grant ANR-17-EURE-0010, ANR Regulia. EP acknowledges support from IUF and ANR Bonsai, grant ANR-23-CE23-0012-01.
{The authors warmly thank M. Carvalho and the anonymous referees for their very useful comments.}




\end{appendices}

\bibliographystyle{plain}
\bibliography{references}



\end{document}